\documentclass[journal]{IEEEtran}
\usepackage{amsthm, amssymb, amsmath}
\usepackage{verbatim, float}
\usepackage{mathrsfs}
\usepackage{graphics}
\usepackage[usenames,dvipsnames]{pstricks}
\usepackage{epsfig}
\usepackage{pst-grad} 
\usepackage{pst-plot} 
\usepackage{cases}
\usepackage{algorithmic}
\usepackage{bm}
\usepackage[noadjust]{cite}
\usepackage[top=1.4cm, bottom=1.4cm, left=1.33cm, right=1.33cm]{geometry}
\usepackage{url}

\usepackage{hhline}

\usepackage{etoolbox}
\makeatletter
\patchcmd{\@makecaption}
  {\scshape}
  {}
  {}
  {}
\makeatother

\usepackage[singlelinecheck=false]{caption}
\usepackage{diagbox}

\theoremstyle{plain}
\newtheorem{theorem}{Theorem}
\newtheorem{lemma}{Lemma}

\newtheorem{corollary}{Corollary}
\newtheorem{proposition}{Proposition}

\theoremstyle{definition}
\newtheorem{definition}{Definition}
\newtheorem{example}{Example}
\newtheorem{remark}{Remark}

\newcommand{\C}{{\mathcal C}}

\renewcommand{\S}{{\mathcal S}}

\newcommand{\Sjjs}{{\mathcal S}_{j}}

\newcommand{\Sjsjs}{{\mathcal S}_{j^*}}

\newcommand{\Saas}{{\mathcal S}_{\bm{\alpha}}}
\newcommand{\Sasas}{{\mathcal S}_{\bm{\alpha}^*}}

\newcommand{\ba}{{\boldsymbol a}}

\newcommand{\bu}{{\boldsymbol u}}
\newcommand{\bv}{{\boldsymbol v}}

\newcommand{\bw}{{\boldsymbol{w}}}

\newcommand{\bbo}{\boldsymbol{\beta}_1}
\newcommand{\bbt}{\boldsymbol{\beta}_2}
\newcommand{\bbl}{\boldsymbol{\beta}_{\ell}}
\newcommand{\bbi}{\boldsymbol{\beta}_i}

\newcommand{\bbm}{\boldsymbol{\beta}_m}
\newcommand{\bbmpo}{\boldsymbol{\beta}_{m+1}}

\newcommand{\bgao}{\boldsymbol{\gamma}_1}
\newcommand{\bgatw}{\boldsymbol{\gamma}_2}
\newcommand{\bgath}{\boldsymbol{\gamma}_3}
\newcommand{\bgaf}{\boldsymbol{\gamma}_4}
\newcommand{\bgafv}{\boldsymbol{\gamma}_5}
\newcommand{\bgas}{\boldsymbol{\gamma}_6}
\newcommand{\bgasv}{\boldsymbol{\gamma}_7}
\newcommand{\bgae}{\boldsymbol{\gamma}_8}

\newcommand{\bgal}{\boldsymbol{\gamma}_{\ell}}
\newcommand{\bgai}{\boldsymbol{\gamma}_i}
\newcommand{\bgaj}{\boldsymbol{\gamma}_j}
\newcommand{\bgam}{\boldsymbol{\gamma}_m}
\newcommand{\bgat}{\boldsymbol{\gamma}_t}
\newcommand{\bgatpo}{\boldsymbol{\gamma}_{t+1}}
\newcommand{\bgampo}{\boldsymbol{\gamma}_{m+1}}

\newcommand{\bc}{{\boldsymbol c}}
\newcommand{\bcv}{{\vec{\boldsymbol c}}}
\newcommand{\bcj}{{\boldsymbol c}_j}
\newcommand{\bcjs}{{\boldsymbol c}_{j^*}}
\newcommand{\bco}{{\boldsymbol c}_1}
\newcommand{\bcn}{{\boldsymbol c}_n}
\newcommand{\bct}{{\boldsymbol c}_2}

\newcommand{\bg}{{\boldsymbol{g}}}
\newcommand{\bgv}{\vec{{\boldsymbol{g}}}}
\newcommand{\bgo}{{\vec{\boldsymbol{g}}^{(1)}}}
\newcommand{\bgt}{{\vec{\boldsymbol{g}}^{(2)}}}
\newcommand{\bgth}{{\vec{\boldsymbol{g}}^{(3)}}}
\newcommand{\bgi}{{\vec{\boldsymbol{g}}^{(i)}}}
\newcommand{\bgl}{{\vec{\boldsymbol{g}}^{(\ell)}}}

\newcommand{\goj}{{\boldsymbol{g}^{(1)}_j}}

\newcommand{\gio}{{\boldsymbol{g}^{(i)}_1}}

\newcommand{\gijs}{{\boldsymbol{g}^{(i)}_{j^*}}}
\newcommand{\gij}{{\boldsymbol{g}^{(i)}_j}}
\newcommand{\gin}{{\boldsymbol{g}^{(i)}_n}}

\newcommand{\glj}{{\boldsymbol{g}^{(\ell)}_j}}

\newcommand{\bbF}{{\mathbb F}}
\newcommand{\fte}{{\mathbb F}_{2^8}}
\newcommand{\ftst}{{\mathbb F}_{2^{16}}}

\newcommand{\bbZ}{{\mathbb Z}}

\newcommand{\ff}{\mathbb{F}}
\newcommand{\ft}{\mathbb{F}_2}
\newcommand{\fq}{\mathbb{F}_q}

\newcommand{\fql}{\mathbb{F}_{q^{\ell}}}

\newcommand{\fqs}{\mathbb{F}_{q^s}}
\newcommand{\fqss}{\mathbb{F}_{q^s}^*}
\newcommand{\fqlt}{\mathbb{F}_{q^{\ell/2}}}
\newcommand{\fqm}{\mathbb{F}_{q^m}}
\newcommand{\fqt}{\mathbb{F}_{q^t}}
\newcommand{\fqmpo}{\mathbb{F}_{q^{m+1}}}

\newcommand{\si}{\sigma}

\newcommand{\bal}{\bm{\alpha}}
\newcommand{\bals}{\bm{\alpha}^*}
\newcommand{\balb}{\overline{\bm{\alpha}}}
\newcommand{\bbe}{\bm{\beta}}

\newcommand{\bga}{\bm{\gamma}}
\newcommand{\bla}{\bm{\lambda}}

\newcommand{\bxi}{\bm{\xi}}
\newcommand{\bzt}{\bm{\zeta}}

\newcommand{\tr}{\mathsf{Tr}}

\newcommand{\trqlqlt}{{\mathsf{Tr}}_{\bbF_{q^\ell}/\bbF_{q^{\ell/2}}}}
\newcommand{\spn}{\mathsf{span}_{\fq}}
\newcommand{\spnt}{\mathsf{span}_{\ft}}

\newcommand{\rank}{\mathsf{rank}_{\fq}}
\newcommand{\rankt}{\mathsf{rank}_{\bbF_2}}
\newcommand{\rankq}{\mathsf{rank}_{\fq}}

\newcommand{\im}{{\sf{im}}}

\newcommand{\define}{\stackrel{\mbox{\tiny $\triangle$}}{=}}

\newcommand{\et}{{\emph{et al.}}}

\newcommand{\Cd}{\mathcal{C}^\perp}
\newcommand{\rsk}{\text{RS}(A,k)}
\newcommand{\grskl}{\text{GRS}(A,k,\vec{\boldsymbol{\lambda}})}
\newcommand{\grsnkl}{\text{GRS}(A,n-k,\vec{\boldsymbol{\lambda}})}

\newcommand{\fa}{f(\boldsymbol{\alpha})}
\newcommand{\faj}{f(\boldsymbol{\alpha}_j)}
\newcommand{\fas}{f(\boldsymbol{\alpha}^*)}
\newcommand{\fab}{f(\overline{\bm{\alpha}})}

\newcommand{\ga}{g(\boldsymbol{\alpha})}
\newcommand{\gaj}{g(\boldsymbol{\alpha}_j)}
\newcommand{\gas}{g(\boldsymbol{\alpha}^*)}

\newcommand{\gi}{g_i}
\newcommand{\gix}{g_i(x)}
\newcommand{\gia}{g_i(\boldsymbol{\alpha})}
\newcommand{\gias}{g_i(\boldsymbol{\alpha}^*)}
\newcommand{\giab}{g_i(\overline{\boldsymbol{\alpha}})}

\newcommand{\gox}{g_1(x)}
\newcommand{\goa}{g_1(\boldsymbol{\alpha})}
\newcommand{\goas}{g_1(\boldsymbol{\alpha}^*)}

\newcommand{\glx}{g_{\ell}(x)}
\newcommand{\gla}{g_{\ell}(\boldsymbol{\alpha})}
\newcommand{\glas}{g_{\ell}(\boldsymbol{\alpha}^*)}

\newcommand{\hi}{h_i}
\newcommand{\hix}{h_i(x)}
\newcommand{\hia}{h_i(\boldsymbol{\alpha})}
\newcommand{\hias}{h_i(\boldsymbol{\alpha}^*)}
\newcommand{\hiab}{h_i(\overline{\boldsymbol{\alpha}})}

\newcommand{\tw}{\tau_W}

\begin{document}

\title{Repairing Reed-Solomon Codes\\ via Subspace Polynomials
\thanks{Part of this work was done when Hoang Dau was with the Coordinated Science Laboratory, University of Illinois at Urbana-Champaign.} 
\thanks{This work was presented in part at the IEEE International Symposium on Information Theory, Aachen, Germany, 2017~\cite{DauMilenkovic2017}.} 
\thanks{
H. Dau is with the Discipline of Computer Science \& Information Technology, School of Science, RMIT University, Australia. Email: sonhoang.dau@rmit.edu.au.
D. T. Xinh is with the Department of Mathematics, Faculty of Natural Science and Technology, Tay Nguyen University, Vietnam. Email: dinhthixinh@ttn.edu.vn.
Han Mao Kiah is with the School of Physical and Mathematical Sciences, Nanyang Technological University, Singapore. Email: hmkiah@ntu.edu.sg.
Tran Thi Luong is with the Department of Information Security, Academy of Cryptographic Technique, Hanoi, Vietnam. Email: luongtranhong@gmail.com. O. Milenkovic is with the Coordinated Science Laboratory, University of Illinois at Urbana-Champaign, USA. Email: milenkov@illinois.edu.}}

\author{Hoang Dau, \emph{Member}, \emph{IEEE}, Dinh Thi Xinh, Han Mao Kiah, \emph{Member}, \emph{IEEE}, Tran Thi Luong, and Olgica Milenkovic, \emph{Fellow}, \emph{IEEE}}

\date{}
\maketitle

\begin{abstract}
We propose new repair schemes for Reed-Solomon codes that use subspace polynomials and hence generalize previous works in the literature that employ trace polynomials. The Reed-Solomon codes are over $\fql$ and have redundancy $r = n-k \geq q^m$, $1\leq m\leq \ell$, where $n$ and $k$ are the code length and dimension, respectively. In particular, for one erasure, we show that our schemes can achieve optimal repair bandwidths whenever $n=q^\ell$ and $r = q^m,$ for all $1 \leq m \leq \ell$.
For two erasures, our schemes use the same bandwidth per erasure as the single erasure schemes, for $\ell/m$ is a power of $q$, and for 
 $\ell=q^a$, $m=q^b-1>1$ ($a \geq b \geq 1$), and for $m\geq \ell/2$ when $\ell$ is even and $q$ is a power of two. 
\end{abstract}

\section{Introduction}
\label{sec:intro}

The \emph{repair bandwidth} is a crucial performance metric of erasure codes when deployed in distributed storage systems~\cite{Dimakis_etal2007, Dimakis_etal2010}.
In such systems, for an underlying finite field $\ff$, e.g. $\ff = \text{GF}(256)$, a data vector in $\ff^k$ is transformed into a codeword vector in $\ff^n$, whose components are 
subsequently stored at different storage nodes. When a node fails, the codedword symbol stored at that node is erased (lost). 
A replacement node (RN) has to recover the content stored at the failed node by downloading relevant information from the remaining operational nodes. 
The repair bandwidth refers to the total amount of information (in bits) that the RN has to download in order to complete the repair process.
If multiple erasures occur, different RNs may also exchange information in a \textit{distributed} manner, and we are interested in the bandwidth used \textit{per erasure}. Alternatively, multiple erasures can be recovered by a \textit{centralized} entity, which, however, is not the focus of this work. 

Reed-Solomon codes~\cite{ReedSolomon1960}, the most practically used maximum distance separable codes~\cite{MW_S}, have been deployed in major distributed storage systems such as the Google File System II, Quantcast File System, Yahoo Object Store, Facebook f4 Storage System, Baidu Atlas Cloud Storage, Backblaze Vaults, and HDFS (see~\cite[Table I]{DauDuursmaKiahMilenkovic2018}). However, they perform poorly as erasure codes under the repair bandwidth metric. For instance, to repair a data chunk of size $256$ MB, the default repair scheme for the Reed-Solomon code $(14,10)$ employed by Facebook's f4~\cite{FBf42014} implementation requires a repair bandwidth of $2.56$ GB. As observed earlier in~\cite{Sathiamoorthy-etal-2013}, the bandwidth used for repairing Reed-Solomon coded data in a Facebook analytics cluster amounts to $10$\%-$20$\% of the total network traffic within the cluster.  

There has been a considerable effort by the research community to improve and optimize the repair bandwidth of Reed-Solomon codes~\cite{Shanmugam2014, GuruswamiWootters2016, GuruswamiWootters2017, YeBarg_ISIT2016, YeBarg_TIT2017, DuursmaDau2017, TamoYeBarg2017, TamoYeBarg2018}.
Several extensions to the case of multiple erasures were also studied~\cite{TamoYeBarg2018,DauDuursmaKiahMilenkovic2018, DauDuursmaKiahMilenkovicTwoErasures2017, BartanWootters2017, MardiaBartanWootters2018, ZhangZhang_ISCIT_2019}. 
The optimal repair bandwidth of Reed-Solomon codes  is generally unknown, except for some full-length codes~\cite{GuruswamiWootters2016, GuruswamiWootters2017, DauMilenkovic2017} and for codes with exponentially large subpacketizations~\cite{TamoYeBarg2017, TamoYeBarg2018}. Constructions of Reed-Solomon codes and repair schemes that trade-off the repair bandwidth and subpacketization size were also investigated in~\cite{ChowdhuryVardy2017, LiWangJafarkhani-Allerton-2017, li2019tradeoff}.
Another line of relevant research has focused on the \textit{I/O cost} (the number of bits \textit{accessed} at helper nodes) of repairing Reed-Solomon codes~\cite{DauDuursmaChu-ISIT-2018, DauViterbo-ITW-2018, LiDauWangJafarkhaniViterbo_ISIT2019}. 
%
%

The focus of this work is on constructions of repair schemes for Reed-Solomon codes over $\fql$ with redundancy $r\geq q^m$ using \textit{subspace polynomials} in $\fql[x]$ whose root sets form $m$-dimensional $\fq$-subspaces of $\fql$ (treated as a vector space over~$\fq$).  
For a \textit{single} erasure, we show that the proposed repair scheme uses a repair bandwidth of $(n-1)(\ell-m)\log_2 q$ bits, which is optimal when $n = q^\ell$ and $r = q^m$ \textit{for all} $1 \leq m\leq \ell$, based on a newly derived lower bound that (slightly) improves upon that in~\cite{GuruswamiWootters2016, GuruswamiWootters2017}. Our scheme generalizes the method introduced in~\cite{GuruswamiWootters2016, GuruswamiWootters2017} which employs \textit{trace polynomials} and only works when $r \geq q^m$ and $(\ell-m)$ \textit{divides} $\ell$. Note that a trace polynomial is a special subspace polynomial and the divisibility condition is imposed by the property that the set of values of the trace polynomial must lie in a \textit{subfield} of $\fql$. This constraint is relaxed in our construction as only \textit{subspaces} are required. 
 
We also develop \textit{distributed} schemes that repair two erasures for Reed-Solomon codes. When $r \geq q^{\ell-1}$, it has been shown by Dau \et~\cite{DauDuursmaKiahMilenkovicTwoErasures2017, DauDuursmaKiahMilenkovic2018} and by Zhang and Zhang~\cite{ZhangZhang_ISCIT_2019} that two erasures can be repaired using a bandwidth of $(n-1)\log_2 q$ bits per erasure via trace polynomial approaches. Our goal is to consider the more general case when $r \geq q^m$, $1~\leq~m~\leq~\ell$. 
We describe constructions of several schemes repairing two erasures for Reed-Solomon codes with a bandwidth of $(n-1)(\ell-m)\log_2 q$ bits per erasure, the same bandwidth as required for a single erasure. More specifically, our constructions apply when $\ell/m$ is a power of $q$, and when $\ell=q^a$ and $m=q^b-1>1$ (for all $a \geq b \geq 1$), and when $m\geq \ell/2$ and $\ell$ is even and $q$ is a power of two. In this setting we also make use of subspace polynomials: However, while any subspace polynomial may be used for the repair scheme of a single erasure, this is not the case for two erasures. There, subspace polynomials satisfying certain additional properties are needed.\vspace{10pt} 

The remainder of the paper is organized as follows. We first provide relevant definitions and terminologies and then discuss the Guruswami-Wootters repair
scheme for Reed-Solomon codes in Section~\ref{sec:pre}. The improved
lower bound on the repair bandwidth of a single erasure is presented in Section~\ref{sec:lower_bound}. We introduce new repair schemes for a single erasure and for two erasures in Section~\ref{sec:single_erasure} and Section~\ref{sec:TwoErasureLinearizedPoly}, respectively. We provide concluding remarks in Section~\ref{sec:conclusion}.

\section{Preliminaries}
\label{sec:pre}

\subsection{Definitions and Notations}
Let $[n]$ denote the set $\{1,2,\ldots,n\}$ and $[m,n]$ the set $\{m,m+1,\ldots,n\}$. Let $\fq$ be the finite field of $q$ elements, for some prime power $q$. Let $\fql$ be an extension field of $\fq$, where $\ell \geq 1$.
We refer to the elements of $\fql$ as \emph{symbols} and the elements of $\fq$ as \emph{subsymbols}. The field $\fql$ may also be viewed as a vector space of dimension $\ell$ over $\fq$, i.e. $\fql \cong \fq^\ell$, and hence each symbol in $\fql$ can be represented as a vector of length $\ell$ over $\fq$. We use $\spn(U)$ to denote the $\fq$-subspace of $\fql$ spanned by a set of elements $U$ of $\fql$. 
We use $\dim_{\fq}(\cdot)$ and $\rankq(\cdot)$ to denote the dimension of a subspace and the rank of a set of vectors over $\fq$, respectively.
The (field) trace of any symbol $\bal \in \fql$ over $\fq$ is defined as 
$\mathsf{Tr}_{\fql / \fq}(\bal) = \sum_{i = 0}^{\ell-1} \bal^{q^i}$. When clear from the context, we omit the subscripts $\fql / \fq$.  

A \emph{linear $[n,k]$ code} $\C$ over $\fql$ is an $\fql$-subspace of $\fql^n$ of dimension $k$. Each element of a code is referred to as a \emph{codeword}. Each element $\bcj$ of a codeword $\vec{\bc}=(\bco,\bct,\ldots,\bcn)\in \fql^n$ is referred to as a codeword symbol. 
The \emph{dual} $\Cd$ of a code $\C$ is the orthogonal complement of $\C$ in $\fql^n$ and has dimension $r = n - k$. 

\begin{definition} 
\label{def:RS}
Let $\fql[x]$ denote the ring of polynomials over $\fql$. A Reed-Solomon code $\rsk \subseteq \fql^n$ of dimension $k$ over a finite
field $\fql$ with evaluation points $A=\{\bal_j\}_{j=1}^n \subseteq \fql$
is defined as 
\[
\rsk = \Big\{\big(f(\bal_1),\ldots,f(\bal_n)\big) \colon f \in \fql[x],\ \deg(f) < k \Big\}. 
\]
The Reed-Solomon code is \emph{full length} if $n = q^\ell$, i.e. $A \equiv \fql$. 
\end{definition}

A \emph{generalized} Reed-Solomon code, $\grskl$, where $\vec{\bla} = (\bla_1,\ldots,\bla_n)\in F^n$, is defined similarly to a Reed-Solomon code, except that the codeword
corresponding to a polynomial $f$ is defined as $\big( \bla_1f(\bal_1),\ldots,\bla_n f(\bal_n) \big)$, where $\bla_j \neq 0$ for all $j \in [n]$. 
It is well known that the dual of a Reed-Solomon code $\rsk$, for any $n \leq |F|$, is a generalized Reed-Solomon code $\grsnkl$, 
for some multiplier vector $\vec{\bla}$~(see~\cite[Chp.~10]{MW_S}). 

Whenever clear from the context, we use $f(x)$ to denote a polynomial of degree at most $k-1$, which corresponds to a codeword of the Reed-Solomon code $\C=\rsk$, and $g(x)$ to denote a polynomial of degree at most $r-1=n-k-1$, which corresponds to a codeword of the dual code $\Cd$. Since
$
\sum_{j=1}^n \gaj(\bla_j\faj) = 0, 
$
we also refer to the polynomial $g(x)$ as a check polynomial for $\C$. Note that when $n = q^\ell$, we have $\bla_j = 1$ for all $j \in [n]$. 
In general, as the column multipliers $\bla_j$ do not play any role in evaluating the repair bandwidth, they are often omitted to simplify the notation (see also Remark~\ref{rm:RS}). 

\subsection{Trace repair framework}
\label{subsec:GW}
First, note that each element of $\fql$ can be recovered from its $\ell$ \textit{independent} traces. More precisely, given a basis $\{\bbe_i\}_{i=1}^\ell$ of $\fql$ over $\fq$, any $\bal \in \fql$ can be uniquely determined given the values of $\tr(\bbe_i\,\bal)$ for $i\in [\ell]$, i.e. $\bal = \sum_{i=1}^\ell\tr(\bbe_i \bal)\bbe^\ast_i$, where $\{\bbe^\ast_i\}_{i=1}^\ell$ is the dual (trace-orthogonal) basis of $\{\bbe_i\}_{i=1}^\ell$ (see, e.g.~\cite[Ch.~2, Def.~2.30]{LidlNiederreiter1986}).

Let $\C$ be an $[n,k]$ linear code over $\fql$ and $\Cd$ its dual. 
If $\bcv = (\bco,\ldots,\bcn) \in \C$ and $\bgv = (\bg_1,\ldots,\bg_n) \in \Cd$ then $\bcv~\cdot~\bgv~= ~\sum_{j=1}^n \bcj\bg_j~=~0$. Suppose $\bcjs$ 
is erased and needs to be recovered.
In the trace repair framework, choose a set of $\ell$ dual codewords $\bgo,\ldots,\bgl$ such that $\dim_{\fq}\big(\{\gijs\}_{i=1}^\ell\big) = \ell$, where $\bgi=(\gio,\ldots,\gin)$, $i \in [\ell]$. Since the trace is a linear map, we obtain the following $\ell$ equations, referred to as the \textit{repair equations}, 
\begin{equation}
\label{eq:repair_equations} 
\tr\big(\gijs \bcjs\big) = -\sum_{j \neq j^*} \tr\big(\gij \bcj\big),\quad i \in [\ell]. 
\end{equation} 
In order to recover $\bcjs$, one needs to retrieve sufficient information from $\{\bcj\}_{j \neq j^*}$ to compute the right-hand sides of \eqref{eq:repair_equations}. 
We define, for all $j \in [n]$,
\begin{equation}
\label{eq:S}
\Sjjs \define \spn\bigg(\left\{\goj,\ldots,\glj\right\}\bigg) 
\end{equation}
and refer to $\Sjjs$ as a \emph{column space} of the repair scheme. The name reflects the fact that $\Sjjs$ is the $\fq$-subspace of $\fql$ spanned by the elements in the $j$-th \textit{column} of the table whose $i$-th row corresponds to the $n$ components of the dual codeword $\bgi$ (see Table~\ref{fig:ConstructionI} for an example).  
Note that $\Sjsjs = \fql$, or equivalently, $\dim_{\fq}(\Sjsjs)=\ell$, which guarantees that $\bcjs$ can be recovered from the $\ell$ (independent) traces $\tr\big(\gijs \bcjs\big)$ on the left-hand side of \eqref{eq:repair_equations}, referred to as the \textit{target traces}.
This is a \textit{necessary and sufficient condition} for the set of dual codewords $\{\bgi\}_{i=1}^\ell$ to form a repair scheme of $\bcjs$.  
 
We now discuss the repair bandwidth of this repair scheme. For $j \neq j^*$, to determine $\tr(\gij \bcj)$ for all $i \in [\ell]$, 
it suffices for the RN to retrieve $b_j=\dim_{\fq}(\Sjjs)$ subsymbols in $\fq$ from the node storing $\bcj$. 
Indeed, suppose $\{\bg^{i_t}_j\}_{t=1}^{b_j}$ is an $\fq$-basis of $\Sjjs$, then by 
retrieving just $b_j$ traces $\{\tr(\bg^{i_t}_j \bcj)\}_{t=1}^{b_j}$
of $\bcj$, referred to as \textit{repair traces}, 
all other traces $\{\tr(\gij \bcj)\}_{i=1}^{\ell}$ can be computed as $\fq$-linear combinations of those $b_j$ traces without any knowledge of $\bcv$. Thus, the RN must download a total of $b = \sum_{j \neq j^*}b_j= \sum_{j \neq j^*} \dim_{\fq}(\Sjjs)$ subsymbols from $\fq$.

We refer to the scheme described above as a \emph{repair scheme based on} $\{\bgi\}_{i=1}^\ell$ using the base field $\fq$. Note that $\fql$ is often referred to as the \textit{coding field}. We omit the base field when it is clear from the context. 
Note that we need one such repair scheme for each codeword symbol $\bcjs$, $j^*\in [n]$. Furthermore, we assume that every node knows all $n$ repair schemes, that is, all $n$ sets of $\ell$ dual codewords, in advance.  
It is known that these types of repair schemes include every possible linear repair scheme for Reed-Solomon codes~\cite{GuruswamiWootters2016}. 
The above discussion is summarized in Lemma~\ref{lem:GW}.  \vspace{5pt}

\begin{lemma}[Guruswami-Wootters~\cite{GuruswamiWootters2016}] 
\label{lem:GW}
Let $\C$ be an $[n,k]$ linear code over $\fql$ and $\Cd$ its dual. 
The repair scheme for $\bcjs$ based on $\ell$ dual codewords $\bgo,\ldots,\bgl$, where $\dim_{\fq}\big(\Sjsjs\big) = \ell$, incurs a repair bandwidth of $\sum_{j \neq j^*} \dim_{\fq}(\Sjjs)$ subsymbols in $\fq$, where $\Sjjs$ are column spaces defined as in \eqref{eq:S}.  
\end{lemma} 

\begin{remark}
\label{rm:RS}
When $\C=\rsk$ is a Reed-Solomon code with $A = \{\bal_j\}_{j=1}^n\in \fql^n$, its dual codewords are of the form $\bgv=(\bla_1 g(\bal_1), \bla_2 g(\bal_2),\ldots,\bla_n g(\bal_n))$, where $g(x)\in \fql[x]$ are polynomials of degrees at most $r-1=n-k-1$ and $\bla_j\in \fql^*$ are fixed column multipliers. A repair scheme for $\bcjs$ is based on $\ell$ polynomials $g_1(x),\ldots,g_\ell(x)$. Since
\[
\Sjjs = \bla_j\spn(\{g_i(\bal_j)\}_{i=1}^{\ell}),
\]
we have $\dim_{\fq}(\Sjjs) = \dim_{\fq}\spn(\{g_i(\bal_j)\}_{i=1}^{\ell})$. Therefore, the multiplier $\bla_j$ are irrelevant for determining the repair bandwidth of the repair scheme based on $g_1(x),\ldots,g_\ell(x)$. We slightly abuse the notation and henceforth ignore $\bla_j$ and referring to $\Sjjs = \spn(\{g_i(\bal_j)\}_{i=1}^{\ell})$ as the column space of the repair scheme for a Reed-Solomon code. Another way to view this simplification is that as recovering $\faj$ is equivalent to recovering $\bla_j\faj$, one can safely ignore $\bla_j$ and focus only on $g_i(x)$ in the construction of low-bandwidth repair schemes for Reed-Solomon codes.  
\end{remark}

\begin{table}[H]
\centering
\tabcolsep=0.1cm
\begin{tabular}{|l||c|c|c|c|c|c|c|c|}
\hline
& $j=1$ & $j=2$ & $j=3$ & $j=4$ & $j=5$ & $j=6$ & $j=7$ & $j=8$\\ 
\hhline{|=||=|=|=|=|=|=|=|=|}
$\bgo$ & $\mathbf{1}$ & $\cdot$ & $\bxi^3$ & $\bxi^6$ & $\bxi$ & $\bxi^5$ & $\bxi^4$ & $\bxi^2$\\
\hline
$\bgt$ & $\bxi$ & $\bxi^4$ & $1$ & $\bxi^2$ & $\bxi^6$ & $\bxi^5$ & $\bxi^3$ & $\cdot$\\
\hline
$\bgth$ & $\bxi^2$ & $\bxi$ & $\bxi^3$ & $1$ & $\bxi^6$ & $\bxi^4$ & $\cdot$ & $\bxi^5$\\
\hhline{|=||=|=|=|=|=|=|=|=|}
$\dim_{\ft}(\S_{j\to 1})$ & $\mathbf{3}$ & $2$ & $2$& $2$& $2$& $2$& $2$& $2$\\
\hline
\end{tabular}
\caption{A list of three dual codewords used to repair the first codeword symbol $\bco$ of an $[8,6]$ Reed-Solomon code over $\mathbb{F}_8$. These dual codewords must be known to all nodes in advance. Here, a dot ``$\cdot$" stands for a zero entry. It suffices for the replacement node to download two bits from each available node, for instance, $\tr(\bxi^4\bct)$ and $\tr(\bxi\bct)$ from the node storing $\bct$, or $\tr(\bxi^2\bc_4)$ and $\tr(\bc_4)$ from the node storing $\bc_4$. Thus, the scheme has a repair bandwidth of $14 = 7\times 2$ bits.}
\label{fig:ConstructionI} 
\end{table}

\begin{example}
\label{ex:trace_repair}
Consider a repair scheme of the first codeword symbol $\bco$ of an $[8,6]$ Reed-Solomon code over $\bbF_8$ that is based on $\ell=3$ dual codewords $\bg^{(1)}$, $\bg^{(2)}$, and $\bg^{(3)}$ as given in Table~\ref{fig:ConstructionI}. 
We have $q = 2$, $n = 8$, $k = 6$, and $j^*=1$ in this case. This repair scheme can be constructed using Construction~I developed in Section~\ref{sec:single_erasure} (see Example~\ref{ex:construction_I}). 
First, this set of three dual codewords corresponds to a repair scheme for $\bco$ because 
\[
\dim_{\ft}(\S_{1}) = \rankt(\{1,\bxi,\bxi^2\})=3=\ell,
\]
where $\bxi$ is a primitive element of $\bbF_8$ satisfying $\bxi^3 + \bxi + 1 = 0$. For instance, for the codeword $\bcv = (?,1,\bxi^2,\bxi^4,0,\bxi,0,0)$ where $\bco$ has been erased, the three repair equations are given below.
\[
\begin{split}
\tr(\bco) &= -\tr(\bxi^3\times\bc_3)-\tr(\bxi^6\times\bc_4)-\tr(\bxi^5\times\bc_6) = 1,\\
\tr(\bxi\bco) &= -\tr(\bxi^4\bct)\hspace{-2pt}-\hspace{-2pt}\tr(\bc_3)\hspace{-2pt}-\hspace{-2pt}\tr(\bxi^2\times\bc_4)\hspace{-2pt}-\hspace{-2pt}\tr(\bxi^5\times\bc_6)\hspace{-2pt} =\hspace{-2pt} 0,\\
\tr(\bxi^2\bco) &= -\tr(\bxi\bct)\hspace{-2pt}-\hspace{-2pt}\tr(\bxi^3\bc_3)\hspace{-2pt}-\hspace{-2pt}\tr(\bc_4)\hspace{-2pt}-\hspace{-2pt}\tr(\bxi^4\times\bc_6) = 0.
\end{split}
\]
As $\{1,\bxi^2,\bxi\}$ is a dual basis of $\{1,\bxi,\bxi^2\}$, we can recover $\bco$ as follows.
\[
\bco = \tr(\bco) + \tr(\bxi\bco)\bxi^2 + \tr(\bxi^2\bco)\bxi = 1.
\]

Next, the repair bandwidth is equal to the sum of the dimensions (over $\ft$) of the column spaces $\S_{j}$, or equivalently, the sum of ranks (over $\ft$) of the entries in columns $j$ of Table~\ref{fig:ConstructionI}, $2 \leq j \leq 8$. 
In this scheme, all column spaces $\S_{j}$, $2 \leq j \leq 8$, have dimension $b_j=2$ over $\ft$. That means the helper node storing $\bcj$, $j \neq 1$, just needs to send two repair traces (two bits) to the replacement node. For example, Node $j = 2$ sends $\tr(\bxi^4 \bct)$ and $\tr(\bxi \bct)$. Nodes $j = 4$ sends $\tr(\bxi^2 \bc_4)$ and $\tr(\bc_4)$. Node $j=4$ does not need to send $\tr(\bxi^6\bc_7)$ because this trace can be written as the sum of the previous two traces, 
\[
\tr(\bxi^6\bc_4)= \tr((\bxi^2+1)\bc_4) = \tr(\bxi^2 \bc_4)+\tr(\bc_4).
\]
Therefore, this scheme has a repair bandwidth of $14 = 7\times 2$ bits.
\end{example} 

\subsection{Definition and Properties of Linearized Polynomials}
 
A linearized polynomial $L(x)$ over $\fql$ is of the form $L(x) = \sum_{i=0}^t \ba_ix^{q^i}$, $\ba_i \in \fql$. A few properties of linearized polynomials needed for our subsequent analysis are discussed below. As $L(\bal+\bbe)=L(\bal)+L(\bbe)$ and $L(c\bal)=cL(\bal)$ for all $\bal, \bbe \in \fql$ and $c \in \fq$, $L(\cdot)$ is an $\fq$-linear map from $\fql$ to itself. For linearized polynomials over $\fq$, i.e. having coefficients in $\fq$, we define the \textit{symbolic multiplication} of $L_1(x)$ and $L_2(x)$ as $L_1\otimes L_2(x) = L_1(L_2(x))$, which is again a linearized polynomial. We use $L^{\otimes s}(x)$ to denote the $s$-th power of $L$ with respect to the symbolic multiplication. The polynomial $l(x) = \sum_{i=0}^t s_ix^i \in \fq[x]$ is called the \textit{conventional $q$-associate} (or just \textit{associate} for short) of $L(x) = \sum_{i=0}^t a_ix^{q^i}\in \fq[x]$. It is known that the associate of $L^{\otimes s}(x)$ is $l^s(x)$ (see, for instance, \cite[Lem.~3.59]{LidlNiederreiter1986}). 

\section{An Improved Lower Bound on the Repair Bandwidth of a Single Erasure}
\label{sec:lower_bound}

In order to evaluate their proposed single erasure repair scheme, Guruswami and Wootters~\cite{GuruswamiWootters2016} established a lower bound on the repair bandwidth for Reed-Solomon codes. 
We start our exposition by improving their bound. The result of this derivation also suggests the number of subsymbols that needs to be downloaded from each available node using an optimal repair scheme. Consequently, the bound allows one to perform a theoretical/numerical search for optimal repair schemes in a simple manner.  

\begin{proposition} 
\label{pro:lower_bound}
Any linear repair scheme for Reed-Solomon codes $\rsk$ over $\fql$ that uses the base field $\fq$ 
requires a bandwidth of at least \vspace{-5pt}
\[
t \lfloor b_{\text{AVE}} \rfloor + (n-1-t)\lceil b_{\text{AVE}} \rceil \vspace{-5pt}
\] 
subsymbols over $\fq$, where $n = |A| \leq q^\ell$, and where $b_{\text{AVE}}$ and $t$
are defined as 
\[
b_{\text{AVE}} \define \log_q\Big(\frac{(n-1)q^\ell}{(r-1)(q^\ell-1)+(n-1)}\Big), \vspace{-5pt}
\]
and $t \define n-1$ if $b_{\text{AVE}} \in \mathbb{Z}$, and \vspace{-5pt}
\[
t \define \left\lfloor \frac{T - (n-1)q^{-\lceil b_{\text{AVE}} \rceil}}{q^{-\lfloor b_{\text{AVE}} \rfloor} - q^{-\lceil b_{\text{AVE}} \rceil}}\right\rfloor
\vspace{-5pt}
\]
otherwise. Here,  \vspace{-5pt}
\[
T \define \frac{(r-1)(q^\ell-1)+(n-1)}{q^\ell}.
\]
\end{proposition} 
\begin{proof} 
The first part of the proof proceeds along the same lines as the proof of~\cite[Thm.~6]{GuruswamiWootters2016}.
But once the optimization problem is solved to arrive at a \emph{fractional} lower bound, rather than allowing
the number of subsymbols downloaded from each available node to be real-valued, we perform a rounding procedure which leads to an improved \emph{integral} lower bound.  
 
Fix an $\bals \in A$ and consider an arbitrary exact linear repair scheme of Reed-Solomon codes for the node storing $\fas$ that uses $b$ subsymbols from $\fq$. By~\cite[Thm.~4]{GuruswamiWootters2016}, there is a set of $\ell$ polynomials $\gox,\ldots,\glx$ such that $\rankq\big(\{\goas,\ldots,\glas\}\big) = \ell$ and $\rankq\big(\{\goa,\ldots,\gla\}\big) = b_{\bal}$, for all $\bal \in A \setminus \{\bals\}$, where
$b = \sum_{\bal \in A \setminus \{\bals\}}b_{\bal}$.  
For each $\bal \in A$, let \vspace{-8pt}
\[
S_{\bal} \define \{\vec{s} = (s_1,\ldots,s_\ell) \in \fq^\ell \colon \sum_{i=1}^\ell s_i\gia = 0\}.\vspace{-7pt}
\]
As $\rank(\{\goa,\ldots,\gla\}) = b_{\bal}$, we have 
$\dim_{\fq}(S_{\bal}) = \ell - b_{\bal}$. Averaging over all \emph{nonzero} vectors $\vec{s} \in \fq^\ell$, 
we obtain \vspace{-8pt}
\begin{multline}
\label{eq:average}
\frac{1}{q^\ell-1} \sum_{\vec{s} \in \fq^\ell \setminus \{\vec{0}\}} 
|\{\bal \in A \setminus \{\bals\} \colon \vec{s} \in S_{\bal}\}|\\
= \frac{1}{q^\ell-1} \sum_{\bal \in A \setminus \{\bals\}} 
|\{\vec{s} \in \fq^\ell \setminus \{\vec{0}\} \colon \vec{s} \in S_{\bal}\}|\\
= \frac{1}{q^\ell-1} \sum_{\bal \in A \setminus \{\bals\}} 
(q^{\ell-b_{\bal}} - 1) =: E.\vspace{-5pt}
\end{multline}
Therefore, there exists some $\vec{s}^* = (s^*_1,\ldots,s^*_\ell) \in \fq^\ell \setminus \{\vec{0}\}$ so that $|\{\bal \colon
\vec{s}^* \in S_{\bal}\}| \geq E$. Let $g(x) \define \sum_{i=1}^\ell s^*_i \gix$. 
By the choice of $\vec{s}^*$, $g(x)$ vanishes on at least $E$ points of $A \setminus \{\bals\}$. Also, since $\vec{s}^* \neq \vec{0}$, $g(\bals) = \sum_{i=1}^\ell s^*_i \gias \neq 0$. Therefore, $g(x)$ corresponds to a nonzero codeword in the dual code $\Cd$ and hence can have at most $r - 1$ roots. Thus,\vspace{-5pt}
\[
\frac{1}{q^\ell-1} \sum_{\bal \in A \setminus \{\bals\}} 
(q^{\ell-b_{\bal}} - 1) = E \leq r-1,\vspace{-5pt}
\] 
or equivalently, 
\begin{equation} 
\label{eq:feasible}
\sum_{\bal \in A \setminus \{\bals\}} q^{-b_{\bal}}
\leq \big((r-1)(q^\ell-1)+(n-1)\big) / q^\ell =: T.\vspace{-5pt}
\end{equation} 
Let \vspace{-10pt}
\begin{equation}
\label{eq:optimization} 
b_{\min} \define \min_{b_{\bal} \in \{0,1,\ldots,\ell\}} \sum_{\bal \in A \setminus \{\bals\}}b_{\bal},\quad\ \text{subject to } \eqref{eq:feasible}.  
\end{equation} 
Then, any feasible repair scheme has to have $b \geq b_{\min}$. To solve the optimization problem~(\ref{eq:optimization}), the authors of~\cite[Thm.~6]{GuruswamiWootters2016} relaxed the condition that $b_{\bal}$ are integer-valued and arrived at a lower bound that reads as $(n-1)b_{\text{AVE}}$, where $b_{\text{AVE}} \define \log_q\big((n-1)/T\big)$.
But one can still solve~\eqref{eq:optimization} for $b_{\bal} \in \{0,1,\ldots,\ell\}$ and arrive at a closed form expression for $b_{\min}$.
To see how to accomplish this analysis, we first let $\{b_1,\ldots,b_{n-1}\}$ refer to $\{b_{\bal} \colon \bal \in A \setminus \{\bals\}\}$. We then 
claim that  \vspace{-5pt}
\[
b^*_1 = \cdots = b^*_t = \lfloor b_{\text{AVE}} \rfloor, 
b^*_{t+1} = \cdots = b^*_{n-1} = \lceil b_{\text{AVE}} \rceil,
\]
where $t$ is the largest integer satisfying $\sum_{i=1}^{n-1}q^{-b^*_i} \leq T$, is an optimal solution of \eqref{eq:optimization}. To this end, if $(b_1,\ldots,b_{n-1})$ is an optimal solution of \eqref{eq:optimization}, and $b_i - b_j \geq 2$ for some $i$ and $j$, we may decrease $b_i$ by one and increase $b_j$ by one, and retain an 
optimal solution. Repeating this ``balancing'' procedure for as many times as possible, we obtain
an optimal solution for which $|b_i - b_j| \leq 1$, $i,j\in [n-1]$.
If $\min_i b_i < \lfloor b_{\text{AVE}} \rfloor$ then $(b_1,\ldots,b_{n-1})$ cannot be a feasible solution. Therefore, $\min_i b_i \geq \lfloor b_{\text{AVE}} \rfloor$. 
Because of the way $t$ was chosen, we always have 
$\sum_{i=1}^{n-1}b_i \geq \sum_{i=1}^{n-1}b^*_i$, 
which establishes the optimality of $(b^*_1,\ldots,b^*_{n-1})$.
Finally, $t$ may be easily computed as follows. 
If $b_{\text{AVE}} \in \mathbb{Z}$ then $t = n-1$, otherwise \vspace{-7pt}
\[
t = \left\lfloor \frac{T - (n-1)q^{-\lceil b_{\text{AVE}} \rceil}}{q^{-\lfloor b_{\text{AVE}} \rfloor} - q^{-\lceil b_{\text{AVE}} \rceil}}\right\rfloor. \qedhere
\]
\end{proof} 

\begin{corollary} 
\label{cr:bound1}
When $n = q^\ell$ and $r = q^m$, for some $m \in [\ell]$, 
any linear repair scheme over the subfield $\fq$ 
of a Reed-Solomon code $\rsk$ defined over $\fql$ requires a bandwidth of at least
$(n-1)(\ell-m)$ subsymbols over $\fq$. 
\end{corollary}
\begin{proof} 
In this case, $b_{\text{AVE}} = \ell-m \in \mathbb{Z}$ and $t = n - 1$, which according to Proposition~\ref{pro:lower_bound} give the desired bound.  
\end{proof}  

Note that the integral bound of Corollary~\ref{cr:bound1} and the Guruswami-Wootters fractional bound coincide. 
However, in many other cases, the integral bound strictly outperforms the fractional
bound. Consider as an example the Facebook RS(14,10) code defined over $\text{GF}(256)$. If the code is repaired over the subfield $\text{GF}(16)$, 
the fractional bound results in at least $28$ downloaded bits, while our integral bound asserts that a download of at least $44$ bits is needed. It is also apparent that the fractional bound does not depend on the base field that the code is repaired over, while the integral bound does. In general, the bigger the order of the base field (given that the coding field is fixed), the larger the gap between the two bounds. 

Also, one may assume that if a repair scheme that achieves the bound of Corollary~\ref{cr:bound1} were to exist, it would require that the 
replacement node downloads $\ell-m$ subsymbols from each available node. This intuitive reasoning will be highly valuable for the design of optimal repair schemes for Reed-Solomon codes. 

\section{Repair One Erasure for Reed-Solomon codes via Subspace Polynomials}
\label{sec:single_erasure}
 
We now describe a way to select dual codewords $\{\bgi\}_{i=1}^\ell$ for $[n,k]$ Reed-Solomon codes with $r = n-k \geq q^m$, $1 \leq m \leq \ell$, so that the corresponding repair scheme for one erasure incurs a low repair bandwidth. The key ingredient of our procedure are subspace polynomials~\cite{MW_S, LidlNiederreiter1986, Goss}. 
When the code is full length, i.e. $n=q^\ell$, and $r = q^m$, the proposed repair scheme indeed achieves the minimum repair bandwidth, according to Proposition~\ref{pro:lower_bound}. 

We henceforth use $f(x)\in \fql[x]$ to denote a polynomial of degree at most $k-1$, which corresponds to a codeword of the Reed-Solomon code $\C=\rsk$, and $g(x)\in \fql[x]$ to denote a polynomial of degree at most $r-1$, which corresponds to a codeword of the dual code $\Cd$. 
In the remainder of this section, we assume that $\fas$ is the erased codeword 
symbol, where $\bals \in A$ is an evaluation point of the code. Applying the trace repair framework to the dual codewords generated by $\{g_i(x)\}_{i=1}^\ell$, we arrive at the following $\ell$ repair equations
\begin{equation}
\label{eq:repair_equations_new} 
\tr\big(\gias\fas\big) = -\hspace{-5pt}\sum_{\bal \in A \setminus\{\bals\}} \tr\big(\gia \fa\big),\quad i \in [\ell]. 
\end{equation} 
Note that in \eqref{eq:repair_equations_new} we ignore the column multipliers to simplify the notation without affecting the bandwidth results of the repair scheme (see Remark~\ref{rm:RS}). 
The column spaces of the repair scheme based on $\{g_i(x)\}_{i=1}^\ell$ are 
\begin{equation}
\label{eq:S_new}
\Saas \define \spn\big(\goa,\ldots,\gla\big), \quad \bal \in A.
\end{equation}
In order to form a repair scheme, it is required that $\dim_{\fq}(\Sasas) = \ell$. Moreover, the scheme uses a bandwidth of $b=\sum_{\bal \in A \setminus \{\bals\}}\dim_{\fq}(\Saas)$ subsymbols in $\fq$.\\

\textbf{Construction I.}
Let $\{\bbo,\ldots,\bbl\}$ be an $\fq$-basis of $\fql$ and $W$ an arbitrary $\fq$-subspace of dimension $m$ of $\fql$. Set $L_W(x) \define \prod_{\bw \in W}(x-\bw)$. The check polynomials used to repair $\fas$ are
\[
g_i(x) \define L_W\big(\bbi(x-\bals)\big)/(x-\bals),\quad i \in [\ell].
\]

Note that $L_W(x)$ constructed as in Construction~I is referred to as a \textit{subspace polynomial}, which is known to be a special type of linearized polynomials over $\fql$ (see~\cite[Ch. 4]{MW_S} and \cite[p. 4]{Goss}). 
The properties of $L_W(x)$ and $g_i(x)$ are captured in Lemma~\ref{lem:linearized}. 
Note that for an $\fq$-linear mapping $L\colon \fql\to \fql$, we use $\ker(L)\define \{\bal \in \fql\colon L(\bal)=0\}$ and $\im(L)=\{L(\bal)\colon \bal\in \fql\}$ to denote the kernel and the image of $L$, respectively. 

\begin{lemma}
\label{lem:linearized}
Let $W$ be an $m$-dimensional $\fq$-subspace of $\fql$ and $L_W(x)$ and $g_i(x)$ defined as in Construction~I. Then the following statements hold. 
\begin{enumerate}
	\item[(a)] $L_W(\cdot)$ is an $\fq$-linear mapping from $\fql$ to itself. Moreover, $\ker(L_W)=W$ and $\dim_{\fq}({\sf im}(L_W)) = \ell-m$.
	\item[(b)] $\deg(g_i) = q^m-1 \leq r-1$ for all $i \in [\ell]$.
	\item[(c)] $\gias = \tw\bbi$ for all $i \in [\ell]$, where $\tw \define \prod_{\bw \in W\setminus \{0\}}\bw$. 
	\item[(d)] $\dim_{\fq}(\Sasas) = \ell$.
	\item[(e)] $\dim_{\fq}(\Saas) \leq \dim_{\fq}({\sf im}(L_W)) = \ell-m$ for $\bal \neq \bals$. 
\end{enumerate}
\end{lemma} 
\begin{proof}[Proof of Lemma~\ref{lem:linearized}] 
A proof of Part (a) can be found in~\cite[p.~4]{Goss}.
As $|W| = q^m$, $L_W(x)$ has degree $q^m$, which implies Part~(b). To prove (c), 
we first write
\[
L_W(x) = x\prod_{\bw \in W\setminus \{0\}}(x-\bw) = \tw x + x^2M(x),
\]
for some polynomial $M(x)$ of degree $q^m-2$, noting that $(-1)^{q^m-1}=1$ for both even and odd $q$. Therefore,
\[
L_W(\bbi(x-\bals)) = \tw\bbi(x-\bals) + (\bbi(x-\bals))^2 M(x).
\] 
Hence, $\gias = \tw\bbi$, for $i \in [\ell]$. For (d), 
as $\tw \neq 0$ and $\{\bbo,\ldots,\bbl\}$ is an $\fq$-basis of $\fql$, it follows that the set $\{\goas,\ldots,\glas\}$ has rank $\ell$ over $\fq$, or equivalently, $\dim_{\fq}(\Sasas) = \ell$.

It remains to prove Part (e). 
For $\bal \neq \bals$, set $\bgai = \bbi(\bal - \bals)$ so that we have $\gia = \frac{1}{\bal - \bals}L_W\big(\bgai\big)$. Hence, using Part (a),
\[
\begin{split}
\rankq\big(\{\goa,\ldots,\gla\}\big) &= \rankq\big(\left\{L_W(\bgao),\ldots, L_W(\bgal)\right\}\big)\\
&\leq \dim_{\fq}\big({\sf im}(L_W)\big) = \ell - m,
\end{split}
\]
or equivalently, $\dim_{\fq}(\Saas) \leq \dim_{\fq}({\sf im}(L_W)) = \ell-m$ for $\bal~\neq~\bals$.
\end{proof} 

\begin{theorem}  
\label{thm:one_erasure}
Let $q^\ell \geq n \geq r \geq q^m$, where $1\leq m \leq \ell$.
The repair scheme based on check polynomials generated by Construction~I can repair a codeword symbol $\fas$ of an $[n,k]$ Reed-Solomon code using a bandwidth of at most $(n-1)(\ell-m)$ subsymbols in $\fq$. When $n=q^\ell$ and $r = q^m$, this bandwidth is optimal. 
\end{theorem} 
\begin{proof} 
We use the properties of $\gix$ described in Lemma~\ref{lem:linearized} to prove the result. Part (b) of the lemma implies that the polynomials $g_i(x)$ indeed correspond to the dual codewords of the Reed-Solomon code $\C$ with $r = n-k \geq q^m$. Part (c) and Part (d) guarantee that the scheme based on $\{\gix\}_{i=1}^\ell$ can repair $\fas$. Part (e) gives an upper bound on the bandwidths used by the helper nodes that equals
\[
b=\sum_{\bal \in A \setminus \{\bals\}}\dim_{\fq}(\Saas) \leq (n-1)(\ell-m)
\]
subsymbols in $\fq$. When $n=q^\ell$ and $r=q^m$, this matches the lower bound on the repair bandwidth given in Corollary~\ref{cr:bound1}.
\end{proof} 

\begin{example} 
\label{ex:construction_I}
Consider an $[8,6]$ Reed-Solomon code $\rsk$ over $\bbF_8$ with $A \equiv \bbF_8 = \{0,1,\bxi,\bxi^2,\ldots,\bxi^6\}$, where $\bxi$ is a primitive element of $\bbF_8$ satisfying $\bxi^3+\bxi+1=0$. A repair scheme over the base field $\ft$ of the first codeword symbol $\bco=f(0)$, where $f(x) \in \bbF_8[x]$ is a polynomial of degree at most five corresponding to a codeword, can be produced by Construction~I as follows. 
As $r = 8-6= 2^1$, we have $m=1$. Also, $\bals=0$.
\begin{itemize}
	\item Select $\bbo = 1$, $\bbt = \bxi$, and $\bbe_3 = \bxi^2$. 
	\item Select a one-dimensional $\ft$-subspace $W = \{0,1\}$ of $\bbF_8$. 
	\item Let $L_W(x) = (x-0)(x-1)=x(x-1)$.
	\item Set $\gix = L_W(\bbi(x-0))/(x-0) = \bbi^2x-\bbi$, for $i \in [3]$.
\end{itemize}
Thus, the check polynomials used in the repair scheme are $g_1(x) = x-1$, $g_2(x)=\bxi^2x-\bxi$, and $g_3(x) = \bxi^4 x - \bxi^2$.
By evaluating these three polynomials at the elements of $A$, we can obtain three dual codewords as listed in Table~\ref{fig:ConstructionI}. The corresponding repair scheme has a bandwidth of $14 = (8-1)\times (3-1)$ bits, as illustrated in Example~\ref{ex:trace_repair}, which matches the statement of Theorem~\ref{thm:one_erasure}. 
The repair schemes of $\fas$ for $\bals = 1,\bxi,\bxi^2,\ldots,\bxi^6$ can be constructed in a similar manner, using the same $\bbi$'s and $W$. 
\end{example} 


\section{Repairing Two Erasures in Reed-Solomon Codes via Subspace Polynomials}
\label{sec:TwoErasureLinearizedPoly}

We consider an $[n,k]$ Reed-Solomon code $\rsk$ where $r = n - k \geq q^m$ and suppose that two codeword symbols, say $\fas$ and $\fab$, are erased.
A check polynomial $g(x)$ is said to \emph{involve} a codeword symbol $\fa$ if $\ga \neq 0$. 
When only one symbol $\fas$ is erased, every check $g(x)$ that involves $\fas$ can be used to generate a repair equation as follows.
\begin{equation} 
\label{eq:trace_repair}
\tr\big(\gas \fas\big) = -\sum_{\bal \in A \setminus \{\bals\}}\tr\big(\ga\fa\big). 
\end{equation}
However, when two symbols $\fas$ and $\fab$ are erased, in order to recover, say, $\fas$, we no longer have
the freedom to use every possible check that involves $\fas$. Indeed, those checks that involve both $\fas$ and $\fab$
cannot be used in a straightforward manner for repair, because we cannot simply compute the right-hand side sum of \eqref{eq:trace_repair} without retrieving some information from $\fab$, which is not available. 

Following the approach in~\cite{DauDuursmaKiahMilenkovicTwoErasures2017, DauDuursmaKiahMilenkovic2018}, we consider a two-phase repair scheme as follows. In the \textit{Download Phase}, each RN contacts and downloads data from the other $n-2$ available nodes, using $m$ check polynomials involving either $\fas$ or $\fab$ but not both. In this phase, each RN can recover $m$ traces about its erased codeword symbol. The missing $\ell-m$ traces can be recovered in the \textit{Collaboration Phase}, via several rounds of communication in which the two RNs exchange data based on what they have received in the Download Phase to help each other complete the repair process. 
We describe below a \textit{one-round} and a \textit{multi-round} repair schemes for two erasures. Both perform the same steps during the Download Phase but follow different procedures in the Collaboration Phase: the one-round procedure allows the two RNs to exchange all $\ell-m$ missing traces at once while the multi-round allows the missing traces to be recovered batch-by-batch. 
Note that in~\cite{DauDuursmaKiahMilenkovicTwoErasures2017, DauDuursmaKiahMilenkovic2018, ZhangZhang_ISCIT_2019}, only \textit{one} missing trace needs to be recovered in the Collaboration Phase, which is easier to handle. 

\subsection{One-Round Repair Schemes for Two Erasures}
\label{subsec:one_round}

We first show that if there is an $\fq$-subspace $W$ of $\fql$ of dimension $m$ satisfying certain properties then there exists a one-round repair scheme for two erasures with low bandwidth. We then demonstrate that such a subspace always exists when $\ell$ is even, $m\geq \ell/2$, and $q$ is a power of two. 

\begin{theorem} 
\label{thm:first_scheme}
Consider a Reed-Solomon code of full length $n = q^\ell$ and $r \geq q^m$ over $\fql$. 
Suppose there exists an $\fq$-subspace $W$ of $\fql$ of dimension $m$ satisfying
\begin{itemize}
	\item[(P1)] $\tw := \prod_{w \in W \setminus \{0\}} w \in \fq$, and
	\item[(P2)] $L_W^{\otimes 2}(\fql) = \{0\}$, or equivalently, $\im(L_W)\hspace{-2pt}\subseteq\hspace{-2pt} W \hspace{-2pt}=\hspace{-2pt} \ker(L_W)$, where $L_W(x)$ is defined as in Construction~I. 
\end{itemize}
Then there exists a scheme that can repair two arbitrary erasures for this code using a bandwidth of at most $(n-1)(\ell-m)$ subsymbols over $\fq$ per erasure.  
\end{theorem} 

Note that a linearized polynomial $L(x)\in \fql[x]$ satisfying $L^{\otimes 2}(\fql) = \{0\}$, or equivalently, $L(L(x)) \equiv 0\pmod{x^{q^\ell}-x}$, is referred to as a \emph{$2$-nilpotent linearized polynomial} (see~\cite{Reis2018}). However, such a polynomial may not even be a subspace polynomial.
We first modify Construction~I to cope with two erasures. 

\textbf{Construction II.}
Let $W$ be the $m$-dimensional $\fq$-subspace of $\fql$ satisfying (P1) and (P2). Let $\{\bbo,\ldots,\bbm\}$ be an $\fq$-basis of $W/(\balb-\bals)$, where $\bbmpo,\ldots,\bbl$ are chosen so that $\{\bbo,\ldots,\bbl\}$ forms an $\fq$-basis of $\fql$ but are otherwise arbitrary. The check polynomials used to repair $\fas$ and $\fab$ are, respectively,
\[
\begin{split}
g_i(x)& = L_W(\bbi(x-\bals))/(x-\bals),\quad i \in [\ell],\\
h_i(x) &= L_W(\bbi(x-\balb))/(x-\balb),\quad i \in [\ell].
\end{split} 
\]

The following properties of $g_i$ and $h_i$ are needed to show that Construction~II ensures the claimed results. 
\begin{lemma}
\label{lem:constructionII}
Let $W$, $g_i$, $h_i$ be defined as in Construction~II.
Then the following statements hold. 
\begin{itemize}
	\item[(a)]$\deg(g_i(x)) = \deg(h_i(x)) = q^m-1 \leq r-1$, for all $i \in [\ell]$.
	\item[(b)] $g_i(\bals) = h_i(\balb) = \tw\bbi$, for all $i \in [\ell]$.
	\item[(c)] $g_i(\balb) = h_i(\bals) = 0$, for all $i \in [m]$.  
	\item[(d)] $\giab$ and $\hias$ belong to $\spn\{\tw\bbo,\ldots,\tw\bbm\}$, for $i \in [m+1,\ell]$. 
\end{itemize}  
\end{lemma}
Note that (a) implies that $g_i$ and $h_i$ are check polynomials for the code, (b) guarantees that $\{g_i(x)\}_{i=1}^\ell$ and $\{h_i(x)\}_{i=1}^\ell$ form repair schemes for $\fas$ and $\fab$, respectively, while (b) and (c) together imply that $g_1(x),\ldots,g_m(x)$ involve $\fas$ but not $\fab$, and $h_1(x),\ldots,h_m(x)$ involve $\fab$ but not $\fas$. Finally, (d) allows the RNs to help each other using the information obtained from the Download Phase (we will discuss this point in detail later). 
\begin{proof}[Proof of Lemma~\ref{lem:constructionII}]
As $|W| = q^m$, $L_W$ has degree $q^m$, which implies Statement (a). To prove (b), note that $L_W(\bbi(x-\bals)) = \tw\bbi(x-\bals) + (\bbi(x-\bals))^2 M(x)$, for some polynomial $M(x)$, and therefore, $g_i(\bals) = \tw\bbi$. 
The same argument works for $h_i(x)$.  
Next, for $i \in [m]$, as $\bbi(\balb-\bals) \in W$, which is the kernel (or set of roots) 
of $L_W$, it holds that $L_W(\bbi(\balb-\bals)) = 0$, which implies that $g_i(\balb) = 0$. Similarly, $h_i(\bals) = 0$ for all $i \in [m]$, which proves (c). 
Finally, to establish (d), note that (P2) implies that $L_W(\fql) \subseteq \ker(L_W)=W$. Therefore, $\giab = L_W(\bbi(\balb-\bals))/(\balb-\bals) \in W/(\balb-\bals) = \spn\{\bbo,\ldots,\bbm\}=\spn\{\tw\bbo,\ldots,\tw\bbm\}$ as $\tw\in \fq$, for all $i \in [m+1,\ell]$. Similar argument holds for $\hias$, noting that $W/(\balb-\bals) \equiv W/(\bals-\balb)$.   
\end{proof} 
\vspace{-5pt}

\begin{figure}[htb]
\centering
\includegraphics[scale=0.8]{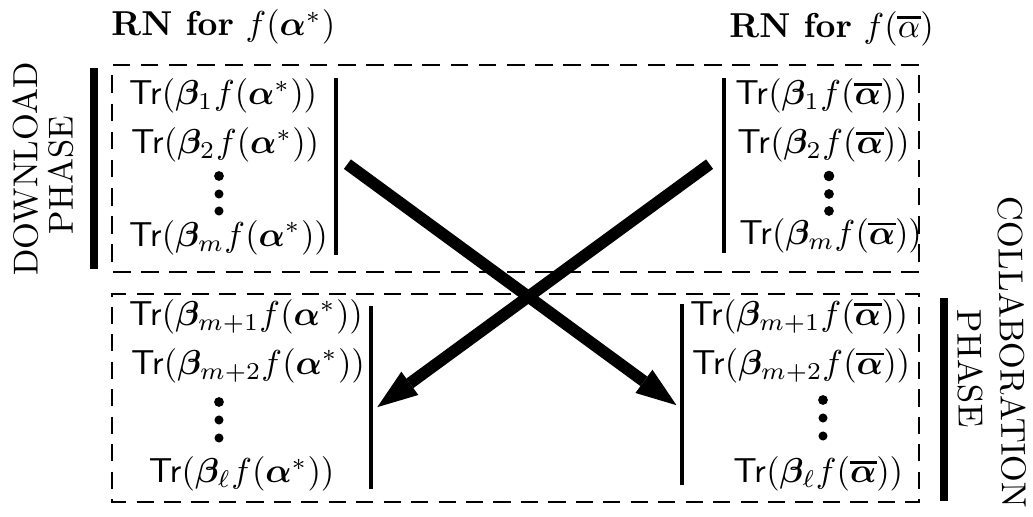}
\caption{Illustration of the two-round repair scheme using check polynomials from Construction~II. The arrows mean the \textit{target} traces available at an RN after the Download Phase can then be used to generate (as linear combinations) the \textit{repair} traces to reconstruct the \textit{target} traces at the other RN in the Collaboration Phase ($\tw$ is ignored for notational simplicity).}
\label{fig:one_round}
\end{figure} 
\vspace{-5pt}

\begin{proof}[Proof of Theorem~\ref{thm:first_scheme}]
We describe below the two phases of the repair scheme based on the two sets of check polynomials $\{g_i(x)\}_{i=1}^\ell$ and $\{h_i(x)\}_{i=1}^\ell$ obtained in Construction~II. 

\textbf{Download Phase.}
In this phase, each RN contacts $n-2$ available nodes to download repair traces and also recover $m$ target traces. 
By Lemma~\ref{lem:linearized}~(e), to obtain all the repair traces $\{\tr\big(\gia \fa\big)\}_{i=1}^\ell$ where $\bal \in A\setminus\{\bals,\balb\}$, the RN for $\fas$ needs to download at most $\ell-m$ traces/subsymbols over $\fq$ from the helper node storing $\fa$.
Hence, it uses a bandwidth of at most $(n-2)(\ell-m)$ subsymbols in this phase.
Next, the RN for $\fas$ uses the first $m$ check polynomials $g_1,\ldots,g_m$, which do not involve $\fab$ according to Lemma~\ref{lem:constructionII}~(c), to construct the following $m$ repair equations.
\[\tr\big(\gias\fas\big) = -\hspace{-5pt}\sum_{\bal \in A \setminus\{\bals,\balb\}} \tr\big(\gia \fa\big),\quad i \in [m]. 
\]
As the result, it can reconstruct $m$ target traces $\tr\big(\gias\fas\big)=\tr\big(\tw\bbi\fas\big)$, $i \in [m]$, of $\fas$. Similarly, the RN for $\fab$ uses the $m$ repair equations
\[ 
\tr\big(\hiab\fab\big) = -\hspace{-5pt}\sum_{\bal \in A \setminus\{\bals,\balb\}} \tr\big(\hia \fa\big),\quad i \in [m], 
\]
to reconstruct $m$ target traces $\tr\big(\hiab\fab\big)=\tr\big(\tw\bbi\fab\big)$, $i \in [m]$, of $\fab$. The bandwidth spent is at most $(n-2)(\ell-m)$ subsymbols over $\fq$.

\textbf{Collaboration Phase.}
In this phase, the two RNs exchange information to help each other recover the last $\ell-m$ missing target traces. To this end, the following two sets of $\ell-m$ repair equations each for $\fas$ based on $g_{m+1},\ldots,g_\ell$ and for $\fab$ based on $h_{m+1},\ldots,h_\ell$ can be used. 
For $i \in [m+1,\ell]$,
\begin{multline}
\label{eq:gl}
\tr\big(\gias\fas\big) + \tr\big(\giab\fab\big)\\ 
= -\sum_{\bal \in A \setminus \{\bals,\overline{\bal}\}} \tr\big(\gia\fa\big).
\end{multline}
\vspace{-15pt}
\begin{multline}
\label{eq:hl}
\tr\big(\hiab\fab\big) + \tr\big(\hias\fas\big)\\ 
= - \sum_{\bal \in A \setminus \{\bals,\overline{\bal}\}} \tr\big(\hia\fa\big).
 \vspace{-10pt}
\end{multline}
It is clear that from the repair traces collected in the
Download Phase that the right-hand-sides of~\eqref{eq:gl} and \eqref{eq:hl} can be determined. 
However, to determine the target traces $\tr\big(\gias\fas\big)$, the RN for $\fas$ also needs to know the repair traces $\tr\big(\giab\fab\big)$, $i \in [m+1,\ell]$.
It turns out that these repair traces can be deduced from the target traces $\tr\big(\hiab\fab\big)$, $i \in [m]$, which are already available at the RN for $\fab$ (see Fig.~\ref{fig:one_round} for an illustration). 
Indeed, by Lemma~\ref{lem:constructionII}~(d), $\giab \in \spn\{\tw\bbo,\ldots,\tw\bbm\}$ for $i \in [m+1,\ell]$. Hence, the RN for $\fab$ can compute $\tr\big(\giab\fab\big)$, $i \in [m+1,\ell]$, as linear combinations of $\tr\big(\tw\bbi\fab\big)=\tr\big(\hiab\fab\big)$, $i \in [m]$, and send these repair traces over to the RN for $\fas$. 
Likewise, the RN for $\fas$ can compute $\tr\big(\hias\fas\big)$, $i \in [m+1,\ell]$, based on $\tr\big(\tw\bbi\fas\big)=\tr(\gias\fas)$, $i\in [m]$, and send these repair traces to the RN for $\fab$. Once all repair traces are available, each RN can recover the missing $\ell-m$ target traces and then the corresponding codeword symbol. 

The bandwidth used in the Collaboration Phase is $\ell-m$ subsymbols over $\fq$ per erasure. Combining with that in the Download Phase, we conclude that this repair scheme incurs a bandwidth of at most $(n-1)(\ell-m)$ subsymbols over $\fq$.
\end{proof}

An important question to ask is whether there exists a subspace $W$ satisfying both properties (P1) and (P2) listed in Theorem~\ref{thm:first_scheme}. In the remainder of this subsection, we establish the existence of such a subspace when $q$ and $\ell$ are even and $m \geq \ell/2$.  

We first describe our key idea. Consider even $q$ and $\ell$, and $m \geq \ell/2$. Note that it is necessary that $m \geq \ell/2$ for (P2) to hold. Indeed, (P2) implies that $\im(L_W) \subseteq \ker(L_W) = W$, which means that $\ell-m = \dim_{\fq}(\im(L_W)) \leq \dim_{\fq}(\ker(L_W)) = m$, or $m\geq \ell/2$.
It seems quite difficult to construct directly a subspace $W$ satisfying both (P1) and (P2). Our strategy is to first construct a subspace satisfying (P1) and then turn it into a subspace satisfying both (P1) and (P2). 
The construction of $W$ is broken down into two steps. 
\begin{itemize}
	\item \textbf{Step 1} Constructing a subspace $U$ of dimension $m-\ell/2$ in $\bbF_{q^{\ell/2}}$ satisfying $\tau_U \in \fq$ (Lemma~\ref{lem:good_tauW}).
	\item \textbf{Step 2} Using $U$ and $\bbF_{q^{\ell/2}}$ to construct $W$, which satisfies both (P1) and (P2) (Lemma~\ref{lem:reduction}). 
\end{itemize}
 
\begin{lemma} 
\label{lem:good_tauW}
For all $m \in [\ell]$, there always exists an $m$-dimensional $\fq$-subspace $U$ of $\fql$ satisfying $\tau_U =\pm 1 \in \fq$.
\end{lemma} 
\begin{proof} 
We prove the following two claims, which together establish Lemma~\ref{lem:good_tauW}. 

\textbf{Claim 1.} For $1 \leq m \leq \ell$ and $s = \text{gcd}(m,\ell)$, there exists an $\fq$-subspace $U_0$ of $\fql$ of dimension $m$ over $\fq$ satisfying $\tau_{U_0} \define \prod_{\bu \in U_0 \setminus \{0\}} \bu = \pm \bzt^{x(q^s-1)}$, for some integer $x$, where $\bzt$ is a primitive element of $\fql$.

To prove Claim~1, take $U_0$ to be an $m/s$-dimensional $\fqs$-subspace of $\fql$. Such a subspace always exists because $m/s \in \bbZ$ and $\ell/s \in \bbZ$. Then $\dim_{\fq}(U_0)=m$. Define a relation ``$\sim$" in $U_0^*\define U_0\setminus \{0\}$ as follows: for $u,v\in U_0^*$, $u \sim v$ if $u/v \in \fqss$. One can verify that this is an equivalence relation and its equivalence classes are multiplicative cosets of $\fqss$, which are of the form $\bu\fqss$, $\bu \in U_0^*$.  
Let $\{\bu_i\fqss \colon i = 1,\ldots,\frac{q^m-1}{q^s-1}\}$ be the set of all disjoint multiplicative cosets of $\fqss$ in $U_0^*$, each of which is of size $q^s-1$. Then, \vspace{-8pt}
\[
U_0^* = \bigcup_{i=1}^{\frac{q^m-1}{q^s-1}}\bu_i\bbF^*_{q^s},\vspace{-5pt}
\] 
which implies that
\[
\tau_{U_0} = \bigg( \prod_{i=1}^{\frac{q^m-1}{q^s-1}} \bu_i \bigg)^{q^s-1} \big( \tau_{\bbF_{q^s}} \big)^{\frac{q^m-1}{q^s-1}}.
\]
Note that $\tau_{\bbF_{q^s}} = -1$. Therefore, choosing an integer $x$ such that $\bzt^x = \prod_{i=1}^{\frac{q^m-1}{q^s-1}} \bu_i$,  we have
$\tau_{U_0} = \pm \bzt^{x(q^s-1)}$, as claimed.

\textbf{Claim 2.} Suppose that there exists an $\fq$-subspace $U_0$ of $\fql$ of dimension $m$ over $\fq$ satisfying $\tau_{U_0} = \pm \bzt^{x(q^s-1)}$ for some integer $x$, where $s = \text{gcd}(m,\ell)$. Then there exists an $m$-dimensional $\fq$-subspace $U$ of $\fql$ satisfying $\tau_U = \pm 1$. 

To prove Claim~2, note that as $\text{gcd}(q^m-1,q^\ell-1)=q^s~-~1$ (based on Euclid's algorithm), there exist integers $y$ and $z$ satisfying 
\[
x(q^s-1) + y(q^m-1) = z(q^\ell-1).
\]
Set $\bga = \bzt^y$ and $U \define \bga U_0$. Then 
\[
\tau_U = \bga^{q^m-1}\tau_{U_0} = \pm \bzt^{y(q^m-1)}\bzt^{x(q^s-1)} = \pm \bzt^{z(q^\ell-1)} = \pm 1. 
\]
Claim 1 and Claim 2 prove Lemma~\ref{lem:good_tauW}. 
\end{proof} 

Lemma~\ref{lem:reduction} is referred to as the \textit{Reduction Lemma} because it reduces the existence of a subspace of $\fql$ satisfying (P1) and (P2) to the existence of a subspace of $\fqlt$ satisfying (P1) only. 

\begin{lemma}[Reduction Lemma] 
\label{lem:reduction}
Suppose that $q$ and $\ell$ are even and $m \in [\ell/2+1,\ell]$. If there exists an $\fq$-subspace $U$ of $\fqlt$ of dimension $m-\ell/2$ satisfying (P1), i.e., $\tau_U \in \fq$, then there exists an $\fq$-subspace $W$ of $\fql$ of dimension $m$ satisfying both (P1) and (P2), i.e., $\tw \in \fq$ and $L_W^{\otimes 2}(\fql) = \{0\}$. 
\end{lemma} 
\begin{proof} 
Denote by $\si$ the trace function $\trqlqlt$, which is an onto function (see~\cite[Thm. 2.23]{LidlNiederreiter1986}). Then $\si(\bv) = \bv^{q^{\ell/2}}+\bv$. As $\fq$ has characteristic two, $\si(\bv) = \bv^{q^{\ell/2}}-\bv = \prod_{\bbe \in \fqlt}(\bv-\bbe)$. This means that $\si(x)$ is also the subspace polynomial of $\fqlt$. Therefore, $\ker(\si) = \fqlt$.

Our proof consists of two steps. First, based on $U$, 
we construct an $\fq$-subspace $V$ of $\fql$ of dimension $m-\ell/2$ such that $V \cap \fqlt = \{0\}$ and $\si(V) = U$. Then, $W = \fqlt \oplus V$ is the desired subspace. Indeed, we have $\dim_{\fq}(W)=\ell/2 + (m-\ell/2) = m$. Moreover, as
\[
W^* = \fqlt^* \cup \Big(\cup_{\bv \in V^*} \big(v + \fqlt\big)\Big),\vspace{-5pt}
\] 
we have 
\[
\begin{split} 
\tw &= \tau_{\fqlt} \prod_{\bv \in V^*}\Big( \prod_{\bbe \in \fqlt}(\bv+\bbe)\Big) = \prod_{\bv \in V^*}\Big( \prod_{\bbe \in \fqlt}(\bv-\bbe)\Big)\\ 
&= \prod_{\bv \in V^*}\si(\bv) = \prod_{\bu \in U^*}\bu = \tau_U \in \fq. 
\end{split} 
\]
Therefore, $W$ satisfies (P1). Note that since $\fq$ has characteristic two, $\tau_{\fqlt}=-1=1$ and $\bv+\bbe = \bv-\bbe$. We now show that (P2) is satisfied as well. For all $\bal \in \fql$, by the definition of $\si$, we have
\[
\begin{split} 
L_W(\bal) &= \prod_{\bv \in V}\prod_{\bbe \in \fqlt}(\bal-(\bv+\bbe))\\ 
&= \prod_{\bv \in V}\prod_{\bbe \in \fqlt}((\bal-\bv)-\bbe)
= \prod_{\bv \in V}\si(\bal-\bv) \in \fqlt,
\end{split} 
\]
where the last equality is due to the fact that $\im(\si) = \fqlt$. Since $W = \fqlt \oplus V \supset \fqlt$, we have $L_W(x) = H(x)\si(x)$, where $H(x) \in \fql[x]$ and $\si(x)$, as defined earlier, is the subspace polynomial of $\fqlt$. Since $L_W(\bal) \in \fqlt$ and $\ker(\si) = \fqlt$, we deduce that 
\[
L_W(L_W(\bal)) = H(L_W(\bal))\si(L_W(\bal)) = H(L_W(\bal))\times 0 = 0. 
\]
Thus, $W$ satisfies (P2) as well. 

We now discuss the construction of $V$, which has dimension $m-\ell/2$ and satisfies $V \cap \fqlt = \{0\}$ and $\si(V)~=~U$. 
Let $\{\bu_j\}_{j=1}^{m-\ell/2}$ be an $\fq$-basis of $U\subseteq \fqlt$. As $\si$ is onto, there exists $m-\ell/2$ elements $\bv_1,\ldots,\bv_{m-\ell/2}$ satisfying $\si(\bv_j) = \bu_j$ for all $j \in [m-\ell/2]$. We claim that the set $\{\bv_j\}_{j=1}^{m-\ell/2}$ is $\fq$-linearly independent. Indeed, suppose there exist $a_1,\ldots, a_{m-\ell/2}$ in $\fq$ so that $0 = \sum_{j=1}^{m-\ell/2} a_j \bv_j$. Applying $\si$ to both sides of this equation, we obtain $0 = \sum_{j=1}^{m-\ell/2} a_j \bu_j$,
which implies that $a_j = 0$ for all $j \in [m-\ell/2]$. 

Set $V = \spn(\{\bv_j\} _{j=1}^{m-\ell/2})$. Then $\dim_{\fq}(V) = m-\ell/2$ and $\si(V) = U$. Moreover, $\fqlt \cap V = \{0\}$. Indeed, as $\si(V)=U$ and $\dim_{\fq}(V)=\dim_{\fq}(U)$, the only element in $V$ that is mapped to $0$ by $\si$ is $0$, while $\si(\fqlt) = \{0\}$. Hence, $\fqlt \cap V = \{0\}$.   
\end{proof} 

Combining Lemma~\ref{lem:good_tauW} and Lemma~\ref{lem:reduction},  we can show that Construction~II works for even $q$ and $\ell$ and $m \geq \ell/2$. 

\begin{corollary} 
\label{cr:II}
Suppose that $2\mid \ell$, $m \geq \ell/2$, and $q=2^s, s\geq 1$. 
Then there exists an $m$-dimensional $\fq$-subspace $W$ of $\fql$ that satisfies the properties (P1) and (P2) in Theorem~\ref{thm:first_scheme}.
Hence, there exists a distributed scheme repairing two erasures for any $[n,k]$ Reed-Solomon code over $\fql$ with $r\hspace{-2pt} =\hspace{-2pt} n\hspace{-2pt}-\hspace{-2pt}k\hspace{-2pt} \geq\hspace{-2pt} q^m$ using a repair bandwidth of at most $(n-1)(\ell-m)$ subsymbols in $\fq$ per erasure. 
\end{corollary} 
\begin{proof} 
If $m = \ell/2$ then we set $W = \fqlt$. Then $\tw = -1 \in \fq$. Hence, $W$ satisfies (P1). Moreover, as shown in the first paragraph in the proof of Lemma~\ref{lem:reduction}, $L_W(x) \equiv \sigma(x) = \trqlqlt(x)$. Therefore, $L_W(\fql) = \fqlt = W = \ker(L_W)$. Equivalently, $L_W(L_W(\fql)) = \{0\}$, which shows that $W$ satisfies (P2). 

Now suppose that $m > \ell/2$. 
By Lemma~\ref{lem:good_tauW}, there exists an $(m-\ell/2)$-dimensional $\fq$-subspace $U$ of $\fqlt$ with $\tau_U \in \fq$. Note that here we replace $m$ by $m-\ell/2$ and $\ell$ by $\ell/2$ in Lemma~\ref{lem:good_tauW}. Then by Lemma~\ref{lem:reduction}, there exists an $\fq$-subspace $W$ of $\fql$ satisfying both (P1) and (P2). 
Applying Theorem~\ref{thm:first_scheme} to $W$, we conclude that there exists a repair scheme for the Reed-Solomon code with the desired bandwidth.  
\end{proof} 

We summarize below the steps to construct a subspace $W$ satisfying both (P1) and (P2) which will prove Theorem~\ref{thm:first_scheme} for $m>\ell/2$. Note that when $m = \ell/2$, we set $W = \fqlt$. 
\begin{itemize}
	\item \textbf{Step 1.} Let $U_0$ be an $\frac{m-\ell/2}{s}$-dimensional $\bbF_{q^s}$-subspace of $\fqlt$, where $s = \text{gcd}(m-\ell/2,\ell/2)$. 
	\item \textbf{Step 2.} Compute $\tau_{U_0} \define \prod_{\bu \in U_0 \setminus \{0\}} \bu = \pm \bzt^{x(q^s-1)}$, where $\bzt$ is a primitive element of $\fqlt$ and $x \in \bbZ$.  
	\item \textbf{Step 3.} Set $U = \bga U_0$, where $\bga = \bzt^y$ and $y\in \bbZ$ such that $x(q^s-1)+y(q^{m-\ell/2}-1)=z(q^{\ell/2}-1)$ for some $z \in \bbZ$. 
	\item \textbf{Step 4.} Let $V$ be an $(m-\ell/2)$-dimensional $\fq$-vector space of $\fql$ constructed as follows. For an $\fq$-basis $\{\bu_i\}_{i=1}^{m-\ell/2}$ of $U$, choose a set $\{\bv_i\}_{i=1}^{m-\ell/2}\in \fql$ so that $\trqlqlt(\bv_i)=\bu_i$ for all $i \in [m-\ell/2]$. Set $V = \spn(\{\bv_i\}_{i=1}^{m-\ell/2})$.
	\item \textbf{Step 5.} Set $W = \fqlt \oplus V$.
\end{itemize}

\begin{example} 
\label{ex:construction_II}
Let $q = 2$, $\ell = 16$, $m=10$, and so $n \geq r \geq 2^{10}$. 
Consider an $[n,k]$ Reed-Solomon code $\rsk$ and suppose that we need to repair two codeword symbols $\fas$ and $\fab$, where $\bals \neq \balb \in A \subseteq \ftst$. 

To construct a distributed repair scheme for both $\fas$ and $\fab$, we first need to construct a $10$-dimensional $\ft$-subspace $W$ of $\ftst$, following the five steps described above.
Note that $m-\ell/2 = 2$, $s = \text{gcd}(2,8) = 2$, $(m-\ell/2)/s = 1$, and $q^s-1=3$.  
\begin{itemize}
	\item \textbf{Step 1.} Let $U_0$ be a 1-dimensional $\bbF_4$-subspace of $\fte$. For instance, $U_0 \equiv \bbF_4$. Note that $\dim_{\ft}(U_0)=2$.
	\item \textbf{Step 2.} In this case, $\tau_{U_0} = \bzt^{0} = \bzt^{0\times 3}$, where $\bzt$ is a primitive element of $\fte$, so that $x = 0$.
	\item \textbf{Step 3.} As $y = 0$, we have $U = \bzt^0U_0 = U_0 = \bbF_4 \subset \fte$. 
	\item \textbf{Step 4.} An $\ft$-basis of $U$ is $\{\bu_1,\bu_2\}=\{1, \bzt^{85}\}$. Next, $V=\spnt\big(\{\bv_1,\bv_2\}\big) = \spnt\big(\{\bxi^{31896}, \bxi^{20312}\}\big)\subset \ftst$, which satisfies $\tr_{\ftst/\fte}(\bv_i)=\bu_i$, $i =1,2$. Here $\bxi$ is a primitive element of $\ftst$. Then $V\cap \fte = \{0\}$.
	\item \textbf{Step 5.} Set $W = \fte \oplus V$. An $\ft$-basis of $W$ can be obtained by combining an $\ft$-basis of $V$ and that of $\fte$, e.g. $\{\bzt^i\}_{i=0}^7$. One can verify that $\tw = 1$. Moreover, $L_W(x)= x^{2^{10}}\hspace{-2pt}+\hspace{-2pt}x^{2^8}\hspace{-2pt}+\hspace{-2pt}x^{2^2}\hspace{-2pt}+\hspace{-2pt}x$. We can verify that $\im(L_W)\subseteq W$ by checking that $L_W(\bbi) \hspace{-2pt}\in\hspace{-2pt} W$ for an $\ft$-basis $\{\bbi\}_{i=1}^{16}$ of $\ftst$.   
\end{itemize}

According to Construction~II, the check polynomials used to repair $\fas$ and $\fab$ are, respectively,
\[
\begin{split}
g_i(x)& = L_W(\bbi(x-\bals))/(x-\bals),\quad i \in [\ell],\\
h_i(x) &= L_W(\bbi(x-\balb))/(x-\balb),\quad i \in [\ell],\vspace{10pt}
\end{split} 
\]
where $\{\bbi\}_{i=1}^{10}$ is an $\ft$-basis of $W/(\bals-\balb)$ and $\{\bbi\}_{i=1}^{16}$ is an $\ft$-basis of $\ftst$. In the Download Phase, since every column space has the same dimension as $\im(L_W)$, which is $\ell-m=6$, each RN downloads six repair traces (bits) from each helper node to recover $m=10$ target traces: $\{\tr(\bbi\fas)\}_{i=1}^{10}$ and $\{\tr(\bbi\fab)\}_{i=1}^{10}$, respectively (we henceforth write $\tr$ instead of $\tr_{\ftst/\ft}$ to avoid notational clutter). This reduces to a scenario corresponding to one erasure because $\gix$ does not involve $\balb$ and $\hix$ does not involve $\bals$, or in other words, $\gias=\hiab=0$, for $i \in [10]$. 

In the Collaboration Phase, the two RNs use $\gix$ and $\hix$, respectively, for $i \in [11,16]$. As it now holds that $g_i(\balb) \neq 0$, the RN for $\fab$ must send six repair traces $\{\tr(\giab\fab)\}_{i=11}^{16}$, which it can determine without the complete knowledge of $\fab$. The reason behind this finding is that as $\giab = L_W(\bbi(\balb-\bals))/(\balb-\bals) \in \im(L_W)/(\balb-\bals) \subseteq W/(\balb-\bals) = \spnt(\{\bbo,\ldots,\bbe_{10})$, one can write $\tr(\giab\fab)$ for $i \in [11,16]$ as a linear combination of $\tr(\bbo\fab),\ldots,\tr(\bbe_{10}\fab)$, which are obtained during the Download Phase. As a result, the RN for $\fas$ now has sufficiently many repair traces to also recover $\tr(\bbi\fas)$ for $i \in [11,16]$. Finally, with $16$ target traces available, it can recover $\fas$. A similar argument works for the process of recovering $\fab$. The bandwidth per erasure is $(n-1)(\ell-m)=6(n-1)$~bits.
\end{example}

\subsection{Multi-Round Repair Schemes for Two Erasures}

Note that Property (P2) of Theorem~\ref{thm:first_scheme} implies that $\im(L_W) \subseteq \ker(L_W) = W$, which means that $\ell-m = \dim_{\fq}(\im(L_W)) \leq \dim_{\fq}(\ker(L_W)) = m$, or $m\geq \ell/2$. 

Next, we develop a repair scheme that also applies for the case that $m < \ell/2$. However, in this case the Collaboration Phase must involve multiple rounds of communications. This makes sense intuitively because when $m$ is small compared to $\ell$, the amount of information (i.e., the number of target traces) each RN knows about its erased codeword symbol after the Download Phase is insufficient to help the other RN recover its content in only one round of communication. 

\begin{theorem}
\label{thm:second_scheme}
Consider a Reed-Solomon code of full length $n = q^\ell$ and $r \geq q^m$ over $\fql$. 
Suppose there exists an $\fq$-subspace $W$ of $\fql$ of dimension $m$ satisfying (with $t = \ell \mod m$)
\begin{itemize}
	\item[(P1)] $\tw := \prod_{w \in W \setminus \{0\}} w \in \fq$,
	\item[(P3)] $\im(L_W^{\otimes \lceil \frac{\ell-m}{m} \rceil}) \cap W$ has dimension at least $t$ over $\fq$,
	\item[(P4)] $\im(L_W^{\otimes \lfloor \frac{\ell-m}{m} \rfloor}) \supseteq W$. 
\end{itemize}
Then there exists a scheme that can repair two arbitrary erasures for this code using a bandwidth of at most $(n-1)(\ell-m)$ subsymbols over $\fq$ per erasure.
\end{theorem}  

Construction~III generates special check polynomials that allow a batch-by-batch reconstruction of traces, the exact meaning of which will be made clear in the description of the Collaboration Phase.

\textbf{Construction III.}
Suppose $W$ is an $\fq$-subspace of $\fql$ of dimension $m$ satisfying (P1), (P3), and (P4).
By (P3), we can select a set of $t$ $\fq$-linearly independent elements $\bgao,\ldots,\bgat$ from the intersection $\im(L_W^{\otimes \lceil \frac{\ell-m}{m} \rceil}) \cap W$. Moreover, due to (P4), we can find $\bgatpo,\ldots,\bgam$ in $W \subseteq \im(L_W^{\otimes \lfloor \frac{\ell-m}{m} \rfloor})$, so that $\{\bgao,\ldots,\bgam\}$ forms an $\fq$-basis of $W$.
By definitions of $\bgao,\ldots,\bgam$, there exist $\bgampo,\ldots,\bgal$ satisfying the \textit{Chain Property} defined as follows. 
\begin{itemize}
	\item $\bgaj = L_W(\bga_{m+j}) = L_W^{\otimes 2}(\bga_{2m+j})=\cdots=L_W^{\otimes \lceil \frac{\ell-m}{m}\rceil}(\bga_{\lceil \frac{\ell-m}{m}\rceil m + j})$, for $j \in [t]$.
	\item $\bgaj = L_W(\bga_{m+j}) = L_W^{\otimes 2}(\bga_{2m+j})=\cdots=L_W^{\otimes \lfloor \frac{\ell-m}{m}\rfloor}(\bga_{\lfloor \frac{\ell-m}{m}\rfloor m + j})$, for $j \in [t+1,m]$.
\end{itemize} 
In other words, $\bga_{m+1},\ldots,\bgal$ are chosen so that $L_W(\bgai) = \bga_{i-m}$ for all $i \in [m+1,\ell]$. 
Finally, we set $\bbi = \bgai/(\balb-\bals)$ for $i \in [\ell]$ and choose the two sets of check polynomials as before.
\[
\begin{split}
g_i(x)& = L_W(\bbi(x-\bals))/(x-\bals),\quad i \in [\ell],\\
h_i(x) &= L_W(\bbi(x-\balb))/(x-\balb),\quad i \in [\ell].
\end{split} 
\]
\\
For example, when $\ell=8$, $m = 3$, $t = 2$, we have\vspace{-5pt}
\[
\begin{split}
\bgao &= L_W(\bga_4) = L_W^{\otimes 2}(\bga_7),\\
\bga_2 &= L_W(\bga_5) = L_W^{\otimes 2}(\bga_8),\\
\bga_3 &= L_W(\bga_6).
\end{split}
\]

\begin{lemma} 
\label{lem:constructionIII}
Let $W$, $\bga_i$, $\bbi$, $g_i$, $h_i$ be defined as in Construction~III.
Then the following statements hold.
\begin{itemize}
	\item[(a)]$\deg(g_i(x)) = \deg(h_i(x)) = q^m-1 \leq r-1$, for all $i \in [\ell]$.
	\item[(b)] $g_i(\bals) = h_i(\balb) = \tw\bbi$, for all $i \in [\ell]$.
	\item[(c)] $g_i(\balb) = h_i(\bals) = 0$, for all $i \in [m]$. 
	\item[(d)] $\{\bgai\}_{i=1}^\ell$ and $\{\bbi\}_{i=1}^\ell$ are $\fq$-bases of $\fql$. 
	\item[(e)] $\giab = \hias = \bbe_{i-m}$ for $i \in [m+1,\ell]$.
\end{itemize}
\end{lemma} 
\begin{proof}[Proof of Lemma~\ref{lem:constructionIII}]
The first three statements follow in the same way as those in the proof of Lemma~\ref{lem:constructionII}. For (d) to hold, it suffices to show that the set $\{\bgai\}_{i=1}^\ell$ is $\fq$-linearly independent. We prove this by induction. 

First, $\bgao,\ldots,\bgam$ are $\fq$-linearly independent by definition. Suppose that $\bgao,\ldots,\bga_s$, where $s \in [m,\ell-1]$, are linearly independent. We aim to show that $\bgao,\ldots,\bga_{s+1}$ are also $\fq$-linearly independent. 
Assume that we can write $\bga_{s+1} = \sum_{i=1}^s a_i \bga_i$, for some $a_i \in \fq$. Applying $L_W$ to both sides of this equation and noting that $L_W$ is $\fq$-linear, we have
\[
L_W(\bga_{s+1}) - \sum_{i=1}^s a_i L_W(\bga_i) = 0.
\]  
Note that $L_W(\bgai)=0$ for $i \in [m]$ as such $\bgai$ belongs to $W$ and moreover, $L_W(\bgai) = \bga_{i-m}$ for $i \in [m+1,\ell]$. Hence, the above equation implies that there exists a nontrivial $\fq$-linear combination of $\bgao,\ldots,\bga_s$ equal to zero, which contradicts our induction hypothesis. Therefore, $\bgao,\ldots,\bga_{s+1}$ must also be $\fq$-linearly independent. 

Finally, to prove (e), again using the fact that $L_W(\bgai) = \bga_{i-m}$, for $i \in [m+1,\ell]$, we have
\[
\gi(\balb)=\frac{L_W(\bbi(\balb-\bals))}{\balb-\bals} = \frac{L_W(\bgai)}{\balb-\bals}=\frac{\bga_{i-m}}{\balb-\bals} = \bbe_{i-m}. 
\] 
Similarly, 
\[
\hi(\bals)=\frac{L_W(\bbi(\bals-\balb))}{\bals-\balb} = \frac{L_W(-\bgai)}{\bals-\balb}=\frac{\bga_{i-m}}{\balb-\bals} = \bbe_{i-m}. 
\] 
This completes the proof of Lemma~\ref{lem:constructionIII}. 
\end{proof} 

\begin{figure}[htb]
\centering
\includegraphics[scale=0.8]{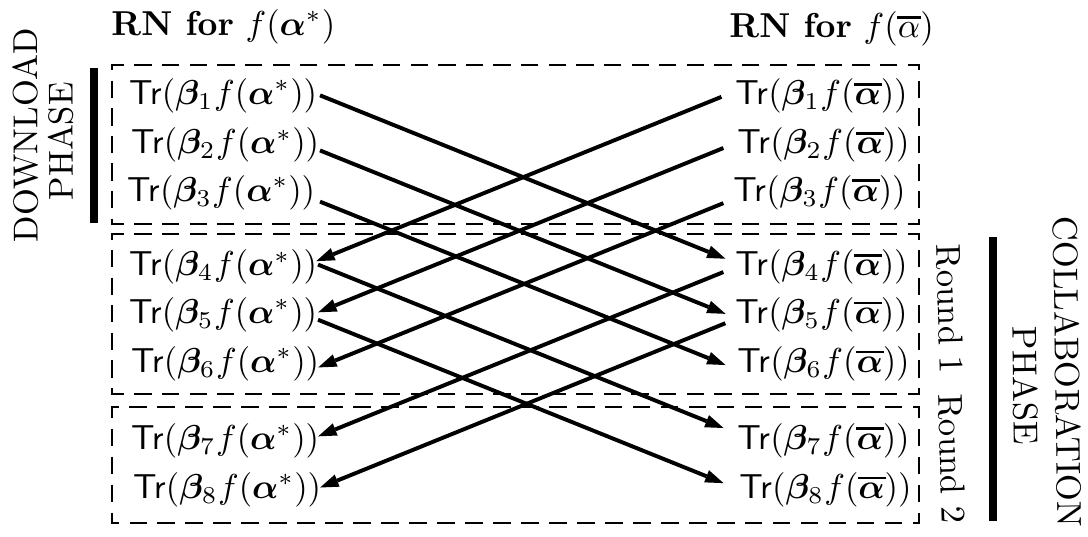}
\caption{Illustration of the process of repairing traces batch-by-batch using check polynomials from Construction~III. An arrow means the target trace available at an RN is then used as a repair trace to reconstruct a target trace at the other RN ($\tw$ is ignored).}
\label{fig:batch}
\end{figure}

\begin{proof}[Proof of Theorem~\ref{thm:second_scheme}]
We use a two-phase repair scheme based on the check polynomials produced by Construction~III that can repair $\fas$ and $\fab$ with a bandwidth of $(n-1)(\ell-m)$ subsymbols per erasure. The Download Phase is the same as that in the proof of Theorem~\ref{thm:first_scheme} and is hence omitted. 
We note that when the Download Phase is completed, 
\begin{itemize}
	\item the RN for $\fas$ has obtained $m$ \textit{target} traces of $\fas$: $\tr\big(\gias\fas\big)=\tr\big(\tw\bbi\fas\big)$, $i \in [m]$, 
	\item the RN for $\fab$ has obtained $m$ \textit{target} traces of $\fab$: $\tr\big(\hiab\fab\big)=\tr\big(\tw\bbi\fab\big)$, $i \in [m]$.
\end{itemize}
Each RN has used a bandwidth of $(n-2)(\ell-m)$ subsymbols. 

The Collaboration Phase consists of $\lceil \frac{\ell-m}{m}\rceil$ rounds. 
In the first round, by Lemma~\ref{lem:constructionIII}~(e), based on the target traces obtained in the Download Phase, the two RNs can construct and exchange the following $m$ repair traces $(i \in [m+1,2m])$ 
\[
\begin{split}
\tr(\giab\fab)\hspace{-2pt}&=\hspace{-2pt}\tr(\bbe_{i-m}\fab)\hspace{-2pt}=\hspace{-2pt}\tr(h_{i-m}(\balb)\fab)/\tw\\
\tr(\hias\fas)\hspace{-2pt}&=\hspace{-2pt}\tr(\bbe_{i-m}\fas)\hspace{-2pt}=\hspace{-2pt}\tr(g_{i-m}(\bals)\fas)/\tw,
\end{split}\vspace{5pt}
\] 
which subsequently enable them to determine the target traces $\tr(\bbi\fas)$ and $\tr(\bbi\fab)$, $i\in[m+1,2m]$, respectively, using the corresponding repair equations. 
Subsequent rounds are carried out in a similar manner, each of which allows each RN to construct and exchange $m$ repair traces based on the target traces recovered in the previous round. These repair traces in turn will allow the RNs to recover a batch of new $m$ target traces. An exception is when $t = \ell \mod m > 0$, as the batch of traces recovered in the \textit{last round} consists of $t$ traces instead of $m$. 
\end{proof} 

Continuing the example for $\ell=8, m=3$, and $t = 2$, the traces repaired in the two phases are illustrated in Fig.~\ref{fig:batch}. 
We now describe two sets of parameters $\ell$ and $m$ for which Construction~III is feasible. 

\begin{corollary} 
\label{cr:III_1}
Suppose that $\ell/m\in \bbZ$ is a power of $q$. 
Then there exists an $m$-dimensional $\fq$-subspace $W$ of $\fql$ that satisfies the properties (P1), (P3), and (P4) in Theorem~\ref{thm:second_scheme}. Moreover, $L_W(x) = 
x^{q^m}-x$. 
Hence, there exists a scheme repairing two erasures for full-length Reed-Solomon code over $\fql$ with $r \geq q^m$ when $\ell/m$ is a power of $q$. The required bandwidth is at most $(n-1)(\ell-m)$ subsymbols per erasure. 
\end{corollary} 
\begin{proof} 
Set $W = \fqm$. Since $m \mid \ell$, $W$ is an $\fq$-subspace of $\fql$ and has dimension $m$ over $\fq$. Then $L_W(x) = x^{q^m}-x$ (Fermat's theorem, see, for example,~\cite[Ch. 4, Cor. 3]{MW_S}). Hence, $\tw = -1 \in \fq$, i.e., $W$ satisfies (P1). 
Since $t = \ell \mod m = 0$, (P3) is trivially satisfied. It remains to show that (P4) holds.
As introduced in Section~\ref{sec:pre}, let $l(x) = x^m-1$ be the associate of $L_W(x)$. Then $l^{ \frac{\ell-m}{m}}(x)$ is the associate of $L_W^{\otimes \frac{\ell-m}{m} }(x)$.
Using the assumption that $\ell/m$ is a power of $q$, we have 
\[
\begin{split}
l^{ \frac{\ell-m}{m} }(x)
 &= (x^m-1)^{\ell/m-1} = \frac{(x^m-1)^{\ell/m}}{x^m-1}
 = \frac{x^{\ell}-1}{x^m-1}\\ &= \sum_{i=0}^{\ell/m-1}x^{mi}. 
 \end{split}\vspace{-5pt}
\]
Therefore, \vspace{-5pt}
\[
L_W^{\otimes \frac{\ell-m}{m} }(x)
= \sum_{i=0}^{\ell/m-1}x^{q^{mi}}
= \sum_{i=0}^{\ell/m-1}x^{(q^m)^i} = \tr_{\fql/\fqm}(x). 
\]
As the trace function is onto (see, for instance,~\cite[Thm. 2.23]{LidlNiederreiter1986}), 
we have $\im\big(L_W^{\otimes \frac{\ell-m}{m} }\big) = \fqm = W$. Thus, (P4) is satisfied. 
\end{proof} 
	
\begin{corollary} 
\label{cr:III_2}
Suppose that $\ell=q^a$ and $m=q^b-1>1$ for some $a \geq b \geq 1$. 
Then there exists an $m$-dimensional $\fq$-subspace $W$ of $\fql$ that satisfies the properties (P1), (P3), and (P4) in Theorem~\ref{thm:second_scheme}. Moreover, $L_W(x) = 
\tr_{\fqmpo/\fq}(x) = \sum_{i = 0}^{m}x^{q^i}$. 
Hence, there exists a scheme repairing two erasures for full-length Reed-Solomon code over $\fql$ with $r \geq q^m$ when $\ell=q^a$ and $m=q^b-1$ for some $1 \leq b \leq a$. The required bandwidth is at most $(n-1)(\ell-m)$ subsymbols per erasure. 
\end{corollary} 

We need an auxiliary result for the proof of Corollary~\ref{cr:III_2}. 

\begin{lemma} 
\label{lem:composition_trace}
Suppose that $\ell=q^a$ and $m=q^b-1>1$ for some $a \geq b \geq 1$. Let $t = \ell \mod m$. Then
\begin{equation} 
\label{eq:trace_exp}
\tr_{\fqmpo/\fq}^{\otimes \lceil \frac{\ell-m}{m} \rceil}(x) = \tr_{\fql/\fqt}(x).
\end{equation} 
\end{lemma} 
\begin{proof} 
Note that the associates of $\tr_{\fqmpo/\fq}(x)$ and $\tr_{\fql/\fqt}(x)$ are $p_1(x) = \sum_{i=0}^m x^i$ and $p_2(x) = \sum_{i=0}^{\ell/t-1}x^{it}$, respectively. 
By~\cite[Lem. 3.59]{LidlNiederreiter1986}, for \eqref{eq:trace_exp}, it suffices to show that $p_1^{\lceil \frac{\ell-m}{m} \rceil}(x) = p_2(x)$.
To this end, first let $c = a \mod b$, so that $t = q^a \mod (q^b-1) = q^c$. Hence,
\[
\left\lceil \frac{\ell-m}{m} \right\rceil = \frac{\ell-t}{m} = \frac{q^c(q^{a-c}-1)}{q^b-1}
=q^c\sum_{j=0}^{\frac{a-c}{b}-1}q^{bj}.
\] 
Therefore, 
\[
\begin{split} 
&p_1^{\lceil \frac{\ell-m}{m} \rceil}(x)
= \Big(\sum_{i=0}^m x^i\Big)^{q^c\sum_{j=0}^{\frac{a-c}{b}-1}q^{bj}}\hspace{-5pt}
= \bigg(\hspace{-5pt}\prod_{j=0}^{\frac{a-c}{b}-1}\hspace{-5pt}\Big( \sum_{i=0}^m x^{iq^{bj}} \Big)\bigg)^{q^c}\\
&=\bigg( \sum_{0\leq i_0,i_1,\ldots,i_{\frac{a-c}{b}-1}\leq m}\hspace{-10pt} x^{i_0q^0}x^{i_1q^b}x^{i_2q^{2b}}\cdots x^{i_{\frac{a-c}{b}-1}q^{\left(\frac{a-c}{b}-1\right)b}}\bigg)^t\\
&= \bigg(\sum_{0\leq i_0,i_1,\ldots,i_{\frac{a-c}{b}-1}\leq m} x^{\sum_{j=0}^{\frac{a-c}{b}-1}i_j q^{jb}}\bigg)^t = \bigg(\sum_{i=0}^{\ell/t-1}x^i\bigg)^t\\
&= p_2(x),
\end{split} 
\]
where the second to last equality is due to the fact that the set
\[
\left\{
\sum_{j=0}^{\frac{a-c}{b}-1}i_j q^{jb}\colon 
0\leq i_0,i_1,\ldots,i_{\frac{a-c}{b}-1}\leq m
\right\}
\]
comprises the representations of all integers from $0$ to $\ell/t-1 = q^{a-c}-1$ in base $q^b$. The last equality follows because $t = q^c$.  
\end{proof} 

\begin{proof}[Proof of Corollary~\ref{cr:III_2}]
First, if we set $W$ to be the kernel (in $\fqmpo$) of the trace function of $\fqmpo$ over $\fq$, then by the rank-nullity theorem (see, e.g.~\cite[p. 70]{friedberg2014linear}), $\dim_{\fq}(W)=m$. Note~that as $m+1=q^b$ divides $\ell=q^a$, $W \subset \fqmpo \subseteq \fql$. Clearly, $L_W(x) = \tr_{\fqmpo/\fq}(x)$. Since $\tw$ is the same as the coefficient of $x$ in $L_W(x)$, we deduce that $\tw = 1 \in \fq$. Therefore, $W$ satisfies (P1). 

We now show that (P3) holds. 
By Lemma~\ref{lem:composition_trace}, 
\begin{equation} 
\label{eq:LWTrace}
L_W^{\otimes \lceil \frac{\ell-m}{m} \rceil}(x)
= \tr_{\fqmpo/\fq}^{\otimes \lceil \frac{\ell-m}{m} \rceil}(x) = \tr_{\fql/\fqt}(x).
\end{equation} 
Therefore, since the trace is an onto map, we have 
\[
\im(L_W^{\otimes \lceil \frac{\ell-m}{m} \rceil}) = \fqt.
\]
In order to show that $\im(L_W^{\otimes \lceil \frac{\ell-m}{m} \rceil}) \cap W$ has dimension at least~$t$, it suffices to prove that $W \supset \fqt$. Equivalently, we aim to show that $L_W(x)$ is divisible by $L_{\fqt}(x)=x^{q^t}-x$. By~\cite[Thm. 3.62]{LidlNiederreiter1986}, this holds if and only if the associate $\sum_{i=0}^m x^i$ of $L_W(x)$ is divisible by the associate $x^t-1$ of $L_{\fqt}(x)$. 

To this end, note that $t = q^c$ divides $m+1 = q^b$, and therefore, 
\[
\sum_{i=0}^m x^i = \frac{x^{m+1}-1}{x-1}
= \frac{(x^t-1)\big(\sum_{j=0}^{(m+1)/t-1}x^{tj}\big)}{x-1},
\]
which is divisible by $x^t-1$ because $\sum_{j=0}^{(m+1)/t-1}x^{tj}$ is divisible by $x-1$. The latter holds since over $\fq$ we have  
\[
\sum_{j=0}^{(m+1)/t-1}1^{tj} = (m+1)/t = q^{b-c} = 0.
\]
Thus, (P3) holds. 

As the last step, we demonstrate that $W$ also satisfies (P4).
Our goal is to show that $W \subseteq K \define \im\Big(L_W^{\otimes \lfloor \frac{\ell-m}{m} \rfloor}\Big)$. 
Note that
\[
\dim_{\fq}\hspace{-2pt}\Big(\hspace{-2pt}\ker\Big(L_W^{\otimes \lfloor \frac{\ell-m}{m} \rfloor}\Big)\hspace{-2pt}\Big)
\leq \log_q\hspace{-2pt}\Big(\hspace{-2pt}\deg\Big(L_W^{\otimes \lfloor \frac{\ell-m}{m} \rfloor}\Big)\hspace{-2pt}\Big) = \ell-t-m,
\]
where we used the fact that 
\[
\log_q\hspace{-2pt}\Big(\hspace{-2pt}\deg\Big(L_W^{\otimes \lfloor \frac{\ell-m}{m} \rfloor}\Big)\hspace{-2pt}\Big)\hspace{-2pt} =\hspace{-2pt} m\left\lfloor\hspace{-2pt} \frac{\ell-m}{m}\hspace{-2pt} \right\rfloor \hspace{-2pt}=\hspace{-2pt} m\Big(\frac{\ell-t}{m}-1\Big) \hspace{-2pt}=\hspace{-2pt} \ell-t-m. 
\]
Therefore,
\[
\begin{split}
\dim_{\fq}(K) &= \dim_{\fq}\Big(\im\Big(L_W^{\otimes \lfloor \frac{\ell-m}{m} \rfloor}\Big)\Big)\\ 
&= \ell - \dim_{\fq}\Big(\ker\Big(L_W^{\otimes \lfloor \frac{\ell-m}{m} \rfloor}\Big)\Big)\\ 
&\geq \ell - (\ell-t-m) = t+m,
\end{split}
\]

Now consider the restriction of $L_W$ on $K$. By \eqref{eq:LWTrace}, we have
\[
L_W(K) = L_W\Big(\im\big(L_W^{\otimes \lfloor \frac{\ell-m}{m} \rfloor}\big)\Big) = \im\big(L_W^{\otimes \lceil \frac{\ell-m}{m} \rceil}\big) = \fqt. 
\] 
Using the rank-nullity theorem for $L_W$ restricted on $K$, we obtain
\[
\dim(\ker_K(L_W)) = \dim(K)-\dim(L_W(K)) \geq (t+m)-t = m, 
\]
where $\ker_K(L_W)$ denotes the kernel of $L_W$ restricted to $K$. Since $W = \ker(L_W)$, we deduce that $W = \ker_K(L_W)\subseteq K = \im\Big(L_W^{\otimes \lfloor \frac{\ell-m}{m} \rfloor}\Big)$. Thus, (P4) is satisfied.  
\end{proof} 

\begin{example} 
\label{ex:construction_III}
Consider the previous example (Fig.~\ref{fig:batch}) for $q = 2$, $\ell=8=2^3, m=3=2^2-1$, and $t = \ell \mod m = 2$. 
We illustrate next the key steps of Construction~III. 

As discussed in the proof of Corollary~\ref{cr:III_2}, we set $W$ to be the kernel of the trace $\tr_{\bbF_{2^4}/\ft}$ in $\bbF_{2^4}$, that is,
\[
W = \{0,1,\bxi^{85},\bxi^{170},\bxi^{17},\bxi^{34}, \bxi^{68}, \bxi^{134}\}\subset \bbF_{2^4}, 
\]
where $\bxi$ is a primitive element of $\bbF_{2^8}$. Note that $\bxi^{85}$ and $\bxi^{17}$ are primitive elements of $\bbF_{2^2}$ and $\bbF_{2^4}$, respectively. Moreover, $L_W(x) = \tr_{\bbF_{2^4}/\ft}(x)=x^{2^3}+x^{2^2}+x^2+x$. Also, $L_W^{\otimes 2}(x) =  \tr_{\bbF_{2^8}/{\bbF_{2^2}}}(x)$ and $\im(L_W^{\otimes 2})=\bbF_{2^2}\subset W$.
Therefore, 
\[
\im(L_W^{\otimes 2}) \cap W =  \bbF_{2^2}= \{0,1,\bxi^{85}, \bxi^{170}\}. 
\]
We need to find $\{\bgai\}_{i=1}^8$ before we can determine $\{\bbi\}_{i=1}^8$. Set $\{\bgao,\bgatw\} = \{1,\bxi^{85}\}$, which is a $\ft$-linearly independent set. Next, we pick $\bgath = \bxi^{17}$ to make $\{\bgao,\bgatw,\bgath\}$ an $\ft$-basis of $W$. Then, we set $\bgaf=\bxi^{51}$, $\bgafv = \bxi^7$, $\bgas = \bxi^{13}$, $\bgasv = \bxi^3$, and $\bgae = \bxi^{53}$. It is easy to verify that the Chain Property is indeed satisfied, that is, $L_W(\bgai)=\bga_{i-3}$, for $i \in [4,8]$, or equivalently,
\[
\begin{split}
\bgao &= L_W(\bga_4) = L_W^{\otimes 2}(\bga_7),\\
\bga_2 &= L_W(\bga_5) = L_W^{\otimes 2}(\bga_8),\\
\bga_3 &= L_W(\bga_6).
\end{split}
\]
Finally, we set $\bbi = \bgai/(\balb-\bals)$ for $i \in [8]$, and use these $\bbi$'s and $L_W(x)$ for generating the checks $\{\gix\}_{i=1}^8$ and $\{\hix\}_{i=1}^8$ to repair $\fas$ and $\fab$.
\end{example} 

\section{Conclusions}
\label{sec:conclusion}

We proposed several repair schemes for a single erasure and two erasures in Reed-Solomon codes over $\fql$. Our schemes were constructed using subspace polynomials offering optimal repair bandwidths of $(n-1)(\ell-m)\log_2(q)$ bits for codes of full length $n = q^\ell$ and $r = q^m$, $1 \leq m \leq \ell$, for the case of one erasure. For two simultaneous erasures, our distributed schemes were shown to achieve the same repair bandwidth per erasure for certain ranges of parameters $q$, $\ell$, and $m$. It remains an open problem to construct repair schemes for two or more erasures that work for all possible parameter choices. Another interesting open problem is to further improve the lower bound on the repair bandwidth, which currently appear to be quite loose for $n << q^\ell$. 

\section*{Acknowledgment}

This work has been supported by the 210124 ARC DECRA grant DE180100768 and the NSF grant 1526875.

\bibliographystyle{IEEEtran}
\bibliography{FullLengthTwoErasures}

\end{document}